%% file: icml_2017_arxiv.tex
\documentclass{article}

\usepackage{subfigure}
\usepackage{amsmath}
\usepackage{amssymb}
\usepackage{comment}
\usepackage{graphicx}
\usepackage{cancel}
\usepackage{amsthm}
\usepackage{natbib}
\usepackage{dsfont}
\usepackage{xcolor}
\usepackage{ stmaryrd }
\usepackage{enumitem}

\usepackage{times}
\usepackage{graphicx} 
\usepackage{subfigure}  

\usepackage{natbib}

\usepackage[algo2e,inoutnumbered,linesnumbered,algoruled,vlined]{algorithm2e}

\usepackage{hyperref}

\usepackage[accepted]{icml2017} 

\usepackage{booktabs}

\input{defs}

\usepackage{multibib}
\newcites{New}{References}

\icmltitlerunning{Fractional Langevin Monte Carlo}

\begin{document} 

\twocolumn[
\icmltitle{Fractional Langevin Monte Carlo: Exploring L\'{e}vy Driven Stochastic Differential Equations for Markov Chain Monte Carlo}




\icmlsetsymbol{equal}{*}

\begin{icmlauthorlist}
\icmlauthor{Umut \c Sim\c sekli}{tpt}
\end{icmlauthorlist}

\icmlaffiliation{tpt}{LTCI, T\'{e}l\'{e}com ParisTech, Universit\'{e} Paris-Saclay, 75013, Paris, France}

\icmlcorrespondingauthor{Umut \c Sim\c sekli}{umut.simsekli@telecom-paristech.fr}

\icmlkeywords{levy processes, markov chain monte carlo, langevin dynamics, stochastic gradient mcmc}

\vskip 0.3in
]

\printAffiliationsAndNotice{}

\input{main_content}

\vfill
\pagebreak

\newpage

\onecolumn

\icmltitle{Fractional Langevin Monte Carlo: Exploring L\'{e}vy Driven Stochastic Differential Equations for MCMC\\{\normalsize SUPPLEMENTARY DOCUMENT}}

\icmltitlerunning{Fractional Langevin Monte Carlo -- Supplementary Document}


{
	\centering
	\textbf{Umut \c Sim\c sekli} \\
	LTCI, T\'{e}l\'{e}com ParisTech, Universit\'{e} Paris-Saclay, 75013, Paris, France \\ 
	\url{umut.simsekli@telecom-paristech.fr}

}



\setcounter{section}{0}
\setcounter{equation}{0}
\setcounter{figure}{0}
\setcounter{table}{0}
\setcounter{page}{1}
 \renewcommand{\theequation}{S\arabic{equation}}
 \renewcommand{\thefigure}{S\arabic{figure}}

 \input{supp_content}

\end{document}

%% file: defs.tex
\newcommand{\x}{X}
\newcommand{\y}{Y}
\newcommand{\Y}{{ Y}}
\newcommand{\D}{ {\cal D} }
\newcommand{\ba}[1]{b(#1,\alpha)}
\newcommand{\bta}[1]{\tilde{b}_{h,K}(#1,\alpha)}
\newcommand{\bha}[1]{\hat{b}(#1,\alpha)}
\newcommand{\rmd}{r}

\newtheorem{thm}{Theorem}
\newtheorem{remark}{Remark}
\newtheorem{cor}{Corollary}
\newtheorem{lemma}{Lemma}
\newtheorem{prop}{Proposition}

\DeclareMathOperator*{\argmin}{arg\min}
\DeclareMathOperator*{\argmax}{arg\max}

\newcommand{\sas}{{\cal S} \alpha {\cal S} }

\newcommand{\pb}{{\cal P}^\x}
\newcommand{\pt}{{\cal P}^\y}

\newcommand{\ab}{{\cal A}^\x}
\newcommand{\at}{{\cal A}^\y}

\usepackage{cleveref}

\newtheorem{assumption}{\textbf{H}\hspace{-3pt}}
\Crefname{assumption}{\textbf{H}\hspace{-3pt}}{\textbf{H}\hspace{-3pt}}
\crefname{assumption}{\textbf{H}}{\textbf{H}}

%% file: main_content.tex
\begin{abstract} 

Along with the recent advances in scalable Markov Chain Monte Carlo methods, sampling techniques that are based on Langevin diffusions have started receiving increasing attention. These so called Langevin Monte Carlo (LMC) methods are based on diffusions driven by a Brownian motion, which gives rise to Gaussian proposal distributions in the resulting algorithms. Even though these approaches have proven successful in many applications, their performance can be limited by the light-tailed nature of the Gaussian proposals. In this study, we extend classical LMC and develop a novel Fractional LMC (FLMC) framework that is based on a family of heavy-tailed distributions, called $\alpha$-stable L\'{e}vy distributions. As opposed to classical approaches, the proposed approach can possess large jumps while targeting the correct distribution, which would be beneficial for efficient exploration of the state space. We develop novel computational methods that can scale up to large-scale problems and we provide formal convergence analysis of the proposed scheme. Our experiments support our theory: FLMC can provide superior performance in multi-modal settings, improved convergence rates, and robustness to algorithm parameters.

\end{abstract} 

\section{Introduction}
\label{sec:intro}

Markov Chain Monte Carlo (MCMC) techniques that are based on continuous diffusions have become increasingly popular due to their success in large-scale Bayesian machine learning. In these techniques, the goal is to generate samples from a \emph{target distribution} $\pi$, by forming a continuous diffusion which has $\pi$ as a stationary distribution. 
In practice, $\pi$ is usually known up to a normalization constant, i.e.\ $\pi(\x) \propto \phi(\x) = \exp(-U(\x))$ for $\x \in \mathds{R}^D$, where $\phi$ is called the \emph{unnormalized target} density and $U$ is called the \emph{potential energy} function. 

Originated in statistical physics \cite{RosskyDollFriedman1978}, Langevin Monte Carlo (LMC) is constructed upon the Langevin diffusion that is defined by the following stochastic differential equation (SDE) \cite{Roberts03}: 
\begin{align}
d\x_t = - \nabla U(\x_t) dt + \sqrt{2}dB_t, \label{eqn:langevin_sde}
\end{align}
where $B_t$ denotes the standard $D$-dimensional Brownian motion. Under certain regularity conditions on $U$, the solution process $(\x_t)_{t \geq 0}$ can be shown to be ergodic with $\pi$ \cite{Roberts03}, which allows us to generate samples from $\pi$ by simulating the continuous-time process \eqref{eqn:langevin_sde} in discrete-time. This approach paves the way for the celebrated Unadjusted Langevin Algorithm (ULA) \cite{Roberts03}, that is given as follows:
\begin{align}
\bar{\x}_{n+1} = \bar{\x}_{n} - \eta_{n+1}  \nabla U(\bar{\x}_{n}) + \sqrt{2\eta_{n+1}} \Delta B_{n+1}, \label{eqn:ula}
\end{align}
where $n$ denotes the iterations, $(\eta_n)_n$ is a sequence of step-sizes, and $(\Delta B_{n})_n$ is a sequence of independent and identically-distributed (i.i.d.) standard Gaussian random variables. Convergence properties of ULA have been studied in \cite{lamberton2003recursive}.

In a statistical physics context, $\x_t$ often represents the position of a particle (at time $t$) that is under the influence of a random force. In this case, the Langevin equation \eqref{eqn:langevin_sde} is motivated by the hypothesis that this random force is the sum of many i.i.d.\ random `pulses', whose variance is assumed to be finite \cite{yanovsky2000levy}. Then, by the central limit theorem (CLT), the sum of these pulses converges to a Gaussian random variable, which justify the choice of the Brownian motion in the Langevin equation \eqref{eqn:langevin_sde}.

A natural question arises if we relax the finite variance assumption and allow the random pulses to have infinite variance. In such circumstances, the `usual' CLT would not hold; however, one can still show that the sum of these pulses converges to a broader class of heavy-tailed distributions called \emph{$\alpha$-stable} (or L\'{e}vy-stable) distributions \cite{paul1937theorie}. Since the law of the random force is non-Gaussian in this case, the Brownian motion would not be appropriate in \eqref{eqn:langevin_sde} and it needs to be replaced with the \emph{$\alpha$-stable L\'{e}vy motion}, which will be described in Section~\ref{sec:bg}.

As opposed to the Brownian motion, which is almost surely continuous, the L\'{e}vy motion can contain discontinuities that are often referred to as `jumps'. Due to these jumps, the SDEs that are driven by L\'{e}vy motions are also called \emph{anomalous diffusions}. It has been noticed that this heavy-tailed nature of the L\'{e}vy processes can be more appropriate for modeling natural phenomena that might incur large variations; a situation often encountered in statistical physics \cite{eliazar2003levy}, finance \cite{mandelbrot2013fractals}, and signal processing \cite{kuruoglu1999signal}.

Despite the fact that L\'{e}vy-driven SDEs have been studied in more general Monte Carlo contexts (e.g.\ for financial simulations) \cite{konakov2011weak,mikulevivcius2011rate}, surprisingly, their use in MCMC has been left widely unexplored. In the statistical physics literature, \citet{ditlevsen1999anomalous} considered a L\'{e}vy-driven SDE with a double-well potential and investigated its waiting-times. In a similar context, \citet{eliazar2003levy} developed an approximate technique based on Tauberian theorems for targeting a L\'{e}vy-driven system to a pre-specified distribution, where they required the target distribution to be exactly evaluated. Whilst being relevant, the applicability and the impact of these approaches are rather limited in the domain of machine learning. 

\begin{figure}[t]
\centering
\includegraphics[width=1\columnwidth]{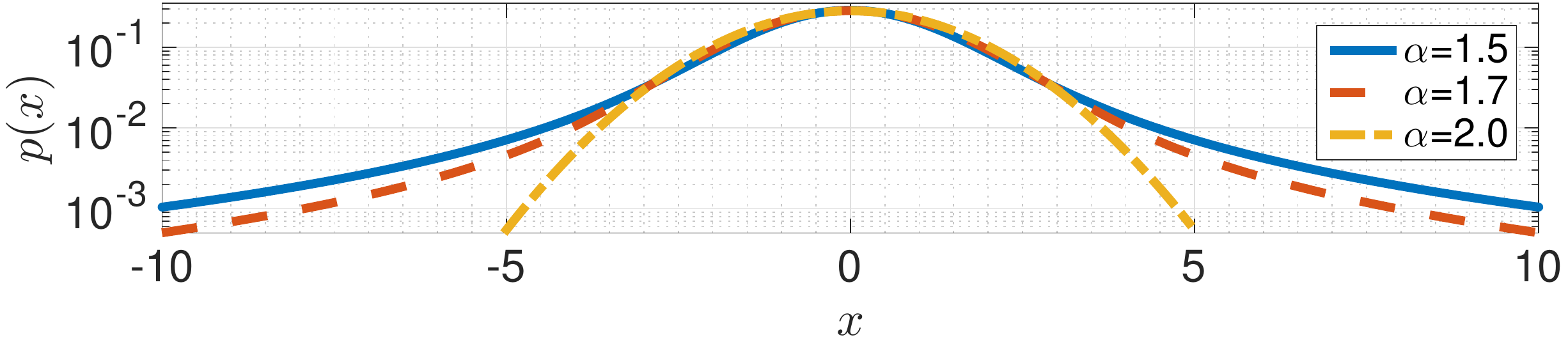} \\
\hfill  \includegraphics[width=0.975\columnwidth]{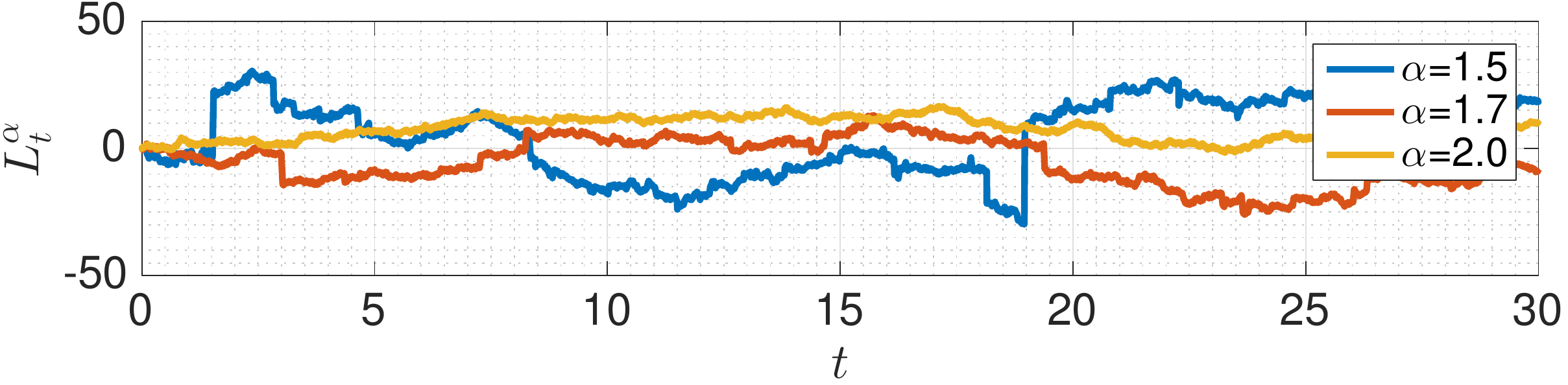}  
\vspace{-12pt}
\caption{Top: probability density functions of $\sas$ . Bottom: illustration of $\alpha$-stable L\'{e}vy motions.}  
\vspace{-12pt}
\label{fig:stablepdf_motion}
\end{figure}
 
In this study, we explore the use of L\'{e}vy-driven SDEs within MCMC. Encouraged by earlier studies that illustrate the benefits of using heavy-tailed distributions in MCMC (e.g., improved convergence rates) \cite{stramer1999langevin,jarner2007convergence}, we aim at investigating the potential benefits of the usage of the L\'{e}vy motions in LMC, in lieu of the classical Brownian motion. 
We extend classical LMC and develop a novel Fractional LMC framework, which targets the correct distribution even in the presence of the jumps induced by the L\'{e}vy motion. We then develop novel computational methods that can scale up to large-scale problems and provide formal theoretical analysis of the convergence behavior of the proposed approach. We support our theoretical results by several synthetic and real experiments. Our experiments show that the proposed approach forms a viable alternative to classical LMC with additional benefits, such as providing superior performance in multi-modal settings, higher convergence rates, and robustness to algorithm parameters. The proposed approach also opens up several interesting future directions, as we will point out in Section~\ref{sec:conc}.

\section{Technical Background}
\label{sec:bg}

\textbf{Stable distributions: }
Stable distributions are heavy-tailed distributions. They are the limiting distributions in the generalized central limit theorem: the properly scaled sum of i.i.d.\ random variables, which do \emph{not} necessarily have finite variance, will converge to an $\alpha$-stable random variable \cite{samorodnitsky1994stable}. 
In this study, we are interested in the centered \emph{symmetric $\alpha$-stable} ($\sas$) distribution, which is a special case of $\alpha$-stable distributions.

The $\sas$ distribution can be seen as a heavy-tailed generalization of the centered Gaussian distribution. The probability density function (pdf) of an $\sas$ distribution cannot be written in closed-form except for certain special cases; however, the characteristic function of the distribution can be written as follows: $X\sim \sas(\sigma) \iff \mathds{E}[\exp(i \omega \x)] = \exp(-|\sigma \omega|^\alpha)$.
Here, $\alpha \in (0,2]$ is called the \emph{characteristic exponent} and determines the tail thickness of the distribution: as $\alpha$ gets smaller, $\sas$ becomes heavier-tailed. 
The parameter $\sigma \in \mathds{R}_+$ is called the \emph{scale} parameter and measures the spread of $X$ around $0$. 

As an important special case of $\sas$, we obtain the Gaussian distribution $\sas(\sigma) = {\cal N}(0,2\sigma^2)$ for $\alpha =2$. 
In Figure~\ref{fig:stablepdf_motion}, we illustrate the (approximately computed) pdf of the symmetric $\alpha$-stable distribution for different values of $\alpha$. As can be clearly observed from the figure, the tails of the distribution vanish quickly when $\alpha =2$ (i.e. Gaussian), whereas the tails get thicker as we decrease $\alpha$.

An important property of the $\alpha$-stable distributions is that their moments can only be defined up to the order $\alpha$, i.e. $\mathds{E}[|X|^p] < \infty$ if and only if $p<\alpha$ for $\alpha \in (0,2)$; implying that $X$ has infinite variance for $\alpha \neq 2$. 
Moreover, even though the pdf of $\sas$ does not admit an analytical form, it is straightforward to draw random samples from stable distributions \cite{chambers1976method}, where efficient implementations are readily available in public software libraries such as the GNU Scientific Library (\url{gnu.org/software/gsl}).

\textbf{SDEs driven by symmetric stable L\'{e}vy processes :}
In this study, we are interested in SDEs driven by symmetric $\alpha$-stable L\'{e}vy processes, which are defined as follows:
\begin{align}
d\x_t =  \ba{\x_{t-}}dt + dL^\alpha_t, \label{eqn:levysde}
\end{align}
where $b$ is called the drift and is chosen as a function of $\alpha$ in our context, and
$\x_{t-}$ will be defined in the sequel. Here, $L^\alpha_t$ denotes the $D$-dimensional $\alpha$-stable L\'{e}vy motion with \emph{independent components}, i.e.\ each component of $L^\alpha_t$ forms an independent \emph{scalar} $\alpha$-stable L\'{e}vy motion, which is defined as follows for $\alpha \in (0,2]$ \cite{duan}:
\begin{enumerate}[label=(\roman*),noitemsep,topsep=0pt,leftmargin=*,align=left]
\item $L_0^\alpha = 0$ almost surely.
\item For $t_0<t_1 < \cdots < t_N$, the increments $ (L_{t_{n}}^\alpha - L_{t_{n-1}}^\alpha )$ are independent ($n = 1,\dots, N$). 
\item The difference $(L_t^\alpha - L_s^\alpha)$ and $L_{t-s}^\alpha$ have the same distribution: $\sas((t-s)^{1/\alpha})$ for $s<t$. 
\item $L_t^\alpha$ has \emph{stochastically continuous} sample paths (i.e.\ continuous in probability): for all $\delta >0$ and $s\geq 0$, $p(|L_t^\alpha - L_s^\alpha| > \delta) \rightarrow 0$ as $t \rightarrow s$.
\end{enumerate}

Due to the stochastic continuity property, $\alpha$-stable L\'{e}vy motions can have a countable number of discontinuities, which are often referred to as `jumps'. As illustrated in Figure~\ref{fig:stablepdf_motion} (bottom), the size of these jumps becomes larger as $\alpha$ get smaller, since $\sas$ becomes heavier tailed.
As a consequence, the sample paths of these processes are continuous from the right and they have left limits at every time \cite{duan}: $\x_{t-}$ hence denotes the left limit of $\x_{t}$ at time $t$.
Therefore, these processes are called c\`{a}dl\`{a}g, i.e.\ the French acronym for \emph{`continue \`{a} droite, limite \`{a} gauche'}.

Similarly to the symmetric $\alpha$-stable distributions, the symmetric $\alpha$-stable L\'{e}vy motions $L_t^\alpha$ coincide with a scaled Brownian motion $\sqrt{2} B_t$ when $\alpha = 2$. This can be simply verified by observing that the difference $L^\alpha_t - L^\alpha_s$ follows a Gaussian distribution ${\cal N}(0,2(t-s))$ and $L_t^\alpha$ becomes almost surely continuous everywhere.

\textbf{Riesz potentials and fractional differentiation: }
Fractional calculus aims to generalize differentiation (and integration) to fractional orders \cite{herrmann2014fractional}. The canonical example of fractional differentiation can be given as the \emph{half-derivative} operator, which coincides with the first-order derivative when applied twice to any function.

In this study, we are interested in \emph{fractional Riesz derivatives} \cite{riesz1949integrale}, which are closely related to $\alpha$-stable distributions. The fractional Riesz derivative directly generalizes the \emph{second-order} differentiation to fractional orders and it is a \emph{non-local} operator. In the one dimensional case, it is defined by the following identity:
\begin{align}
 \D^\gamma f(x) \triangleq {\cal F}^{-1} \{ |\omega|^{\gamma} \hat{f}(\omega) \}, \label{eqn:riesz_def}
\end{align}
where ${\cal F}$ denotes the Fourier transform and $\hat{f}(\omega) = {\cal F} \{ f(x) \}$.
Here, $\gamma > -1$ is the order of the differentiation: for $\gamma \in (-1,0)$ we obtain the Riesz potentials\footnote{For $\gamma<0$, $\D^\gamma$ corresponds to fractional integration. However, we follow the fractional calculus literature and still refer to it as fractional differentiation.}, which will be our main source of interest, and for $\gamma =2$ we obtain the usual second-order differentiation up to a sign difference, i.e.\ $\D^2  f(x) = -d^2 f(x)/ d x^2 $. Note that $\D^1f(x)$ does not coincide with first-order differentiation in general.

\section{Fractional Langevin Monte Carlo}

In this section, we present our main results and construct the proposed Fractional LMC framework step by step. We first develop a L\'{e}vy-driven SDE that targets the correct distribution and analyze the weak convergence properties of its Euler discretization. Afterwards, we develop numerical methods for approximate simulation of the proposed SDE and present formal analysis of the approximation error of the numerical schemes and the weak error analysis of the corresponding Euler discretizations.

In the rest of this paper, we restrict $\alpha$ to be in $(1,2]$ in order the mean of the process to exist. Besides, in all our analyses we focus on the scalar case ($D=1$) for simplicity; however, all our results can be extended for $D>1$. All the proofs are provided in the supplementary document.

\subsection{Invariant measure and weak convergence analysis}
Our first goal is to find a drift $b$ in such a way that the Markov process $(X_t)_{t\geq 0}$ that is a c\`{a}dl\`{a}g solution of the SDE in \eqref{eqn:levysde} would have the target distribution $\pi$ as an invariant distribution. In the following theorem, we present our first main result.
\begin{thm}
\label{thm:ffpe}
Consider the SDE \eqref{eqn:levysde}, where $b$ is defined as: 
\begin{align}
\ba{x} \triangleq \bigl( \D^{\alpha-2} f_\pi(x) \bigr) / \phi(x). \label{eqn:drift}
\end{align}
Here, $f_\pi(x) \triangleq - \phi(x) \partial_x U(x) $ and $\D^\gamma$ is defined in \eqref{eqn:riesz_def}.
Then, $\pi$ is an invariant measure of the Markov process $(\x_t)_{t\geq 0}$ that is a c\`{a}dl\`{a}g solution of the SDE given in \eqref{eqn:levysde}. Furthermore, if $b$ is Lipschitz continuous, then $\pi$ is the unique invariant measure of the process $(\x_t)_{t\geq 0}$.
\end{thm}
The Lipschitz continuity of $\partial_x U$ is a standard condition in LMC for ensuring the uniqueness of the invariant measure, albeit it is often violated in practical applications. In our context, we need $b$ to be Lipschitz continuous for uniqueness, a condition which cannot be easily verified for $\alpha \neq 2$. 
Here, it is also worth noting that when $\alpha \rightarrow 2$, we obtain the classical Langevin diffusion \eqref{eqn:langevin_sde}, as $\lim_{\alpha \rightarrow 2} L_t^\alpha  = \sqrt{2}B_t$ and $\lim_{\alpha \rightarrow 2} \ba{x}  =- \partial_x U(x)$. 

Theorem~\ref{thm:ffpe} suggests that if we could generate continuous sample paths from $(\x_t)_{t\geq 0}$, then we could use them as samples drawn from $\pi$. However, this is not possible since the drift $b$ does not admit an analytical form in general, and even if it could be computed exactly, we still could not simulate the SDE \eqref{eqn:levysde} exactly as it is a continuous-time process.

For now, let us assume that we can exactly compute the drift $b$ and focus on simulating the SDE by considering its Euler-Maruyama discretization \cite{duan,panloup2008recursive}, which is given as follows: 
\begin{align}
\bar{\x}_{n+1} = \bar{\x}_{n} + \eta_{n+1} \ba{\bar{\x}_{n}} + \eta_{n+1}^{1/\alpha} \Delta L^\alpha_{n+1} \label{eqn:euler}
\end{align}
where $n =1,\dots,N$ denotes the time-steps, $N$ is the total number of time-steps (i.e.\ iterations), $(\eta_n)_n$ is a sequence of step-sizes, and $(\Delta L^\alpha_{n})_n$ is a sequence of i.i.d.\ standard symmetric $\alpha$-stable random variables, i.e.\ $\Delta L^\alpha_{n} \sim \sas(1)$. We can clearly observe that this discretization schema is a fractional generalization of ULA given in \eqref{eqn:ula}, where it coincides with ULA when $\alpha = 2$.

The Euler-Maruyama scheme in \eqref{eqn:euler} lets us approximately compute the expectation of a test function $g$ under the target density $\pi$, i.e.\ $\nu(g) \triangleq \int g(\x) \pi(d\x)$, by using sample averages, given as: 
$\bar{\nu}_N(g) \triangleq \frac1{H_N} \sum_{n=1}^N \eta_n g(\bar{\x}_n) $,
where $ H_N = \sum_{n=1}^N \eta_n$. Even though the convergence properties of the estimators obtained via ULA have been well-established \cite{Roberts03,durmus2015non}, it is not clear whether the estimator $\bar{\nu}_N(g)$ converges to the true expectation $\nu(g)$ for $\alpha \neq 2$.

For the convergence analysis, we make use of relatively recent results from the applied probability literature \cite{panloup2008recursive}. In order to establish the convergence of our estimator, we need certain conditions to be satisfied. First, we have a rather standard assumption on the step-sizes:
\begin{assumption}
\label{asmp:stepsize}
$\lim_{n\rightarrow \infty} \eta_n = 0, \quad \lim_{N\rightarrow \infty} H_N  = \infty $, 
\end{assumption}
i.e.\ the step-sizes are required to be decreasing and their sum is required to diverge.
Secondly, we need a more technical Lyapunov condition in order to ensure the stochastic process to be mean-reverting. 
\begin{assumption}
\label{asmp:lyapunov}
Let $V: \mathds{R} \rightarrow \mathds{R}^*_+$ be a function in ${\cal C}^2$, satisfying $\lim_{|x| \rightarrow \infty} V(x) = \infty$, $|\partial_x V| \leq C \sqrt{V}$ for some $C>0$, and $\partial_x^2 V$ is bounded. There exists $a \in (0,1]$, $\delta >0$ and $\beta \in \mathds{R}$, such that $|b|^2 \leq C V^a$ and $ b(\partial_x V) \leq \beta - \delta V^a$, where $b$ is defined in \eqref{eqn:drift}. 
\end{assumption}
Under these conditions, we present the following corollary to Theorem~\ref{thm:ffpe} and \citep[Theorem 2]{panloup2008recursive}, where we establish the weak convergence of the Euler-Maruyama scheme defined in \eqref{eqn:euler}. 
\begin{cor}
\label{cor:euler}
Assume that $b$ is Lipschitz continuous and the conditions \Cref{asmp:stepsize} and \Cref{asmp:lyapunov} hold.  If the test function $g = o\bigl( V^{p/2 + a -1} \bigr)$ with $p\in (0,1/{2}]$, then the following holds: 
\begin{align*}
 \lim_{N \rightarrow \infty} \bar{\nu}_N(g) = \nu(g), \quad \text{ almost surely.}
\end{align*} 
\end{cor}
This corollary shows that under certain regularity and Lyapunov conditions, the Euler-Maruyama scheme in \eqref{eqn:euler} still weakly converges for $\alpha \neq 2$, as long as the drift can be computed exactly. Note that we consider the Lipschitz condition for ensuring the uniqueness of the invariant measure; however, this is not a crucial assumption as one can show that every weak limit of
the sequence $\{\bar{\nu}_N\}_N$ is an invariant probability for the SDE in \eqref{eqn:levysde}.

\subsection{Numerical approximation}

Even though Corollary~\ref{cor:euler} ensures the weak convergence of the Euler scheme, its practical implication is somewhat limited since the Riesz derivatives cannot be computed exactly in general. 
In this section, we develop and analyze numerical methods for approximately computing the drift $b$. 

In \cite{ortigueira2006riesz}, it has been shown that for $\gamma \in (-1,2)$, the Riesz derivative $\D^\gamma$ of a function $f(x)$ can be defined as the limit of the \emph{fractional centered difference} operator $\Delta_h^\gamma$, given as:
$\D^\gamma f(x) = \lim_{h\rightarrow 0} \Delta_h^\gamma f(x)$, where
\begin{align}
  \Delta_h^\gamma f(x) &\triangleq (1/{h^\gamma}) \sum\nolimits_{k=-\infty}^\infty g_{\gamma,k} f(x-kh), \label{eqn:delta_h} 
\end{align}
and $g_{\gamma,k} \triangleq  (-1)^k \Gamma(\gamma+1)/( \Gamma(\frac{\gamma}{2}-k+1 ) \Gamma(\frac{\gamma}{2}+k+1 ) )$. 
By using the above definition, we can rewrite our drift as: 
$\ba{x} = \bigl( \lim_{h\rightarrow 0} \Delta_h^{\alpha-2} f_\pi(x)\bigr)/ \phi(x)$, 
where $f_\pi(x)$ is defined in Theorem~\ref{thm:ffpe}.

We now propose our first numerical scheme for approximating the drift $b$ by following \citet{ccelik2012crank}:
\begin{align}
\ba{x} \approx \bta{x} \triangleq \bigl(\Delta_{h,K}^{\alpha-2} f_\pi(x) \bigr)/ \phi(x),  \label{eqn:approxriesz}
\end{align}
where $\Delta_{h,K}^{\gamma}$ is the \emph{truncated} fractional central difference operator, defined as follows:
\begin{align}
\Delta_{h,K}^{\gamma} f(x) \triangleq (1/{h^\gamma}) \sum\nolimits_{k=-K}^K g_{\gamma,k} f(x-kh).
\end{align}
Here, we merely replaced the Riesz derivative with the central difference operator where we fixed $h$ and truncated the infinite summation in order the numerical scheme to be computationally tractable. 
We provide a numerically stable implementation of \eqref{eqn:approxriesz} in the supplementary document. 

The scheme in \eqref{eqn:approxriesz} provides us a practical way for approximately computing the drift. However, for fixed $h$ and $K$, this approach would yield a certain approximation error and therefore Corollary~\ref{cor:euler} would no longer hold if we replace $b$ by $\tilde{b}$ in \eqref{eqn:euler}. Throughout this section, we analyze this approximation error and the weak error of the Euler scheme with the approximate drift.

 We first analyze the approximation error of our numerical scheme in \eqref{eqn:approxriesz}. Since $\phi(x)$ is constant for a given $x$, we focus on $|\D^\gamma f_\pi(x) - \Delta_{h,K}^{\gamma} f_\pi(x)|$.
Here, we first need a technical regularity condition on $f_\pi$.
\begin{assumption}
\label{asmp:f_reg}
$f_\pi(x) \in {\cal C}^3 (\mathds{R})$ and all derivatives up to order three belong to ${\cal L}_1(\mathds{R})$.
\end{assumption}
We need an additional assumption on $f_\pi$, which ensures the tails of the target distribution $\pi$ vanish sufficiently quickly. 
\begin{assumption}
\label{asmp:tail}
$| f_\pi(x-kh)| \leq C \exp(-|k|h) $ for some $C >0$ and $|k| > K$ for some $K \in \mathds{N}_+$.
\end{assumption}
Now, we present our second main result.
\begin{thm}
\label{thm:riesz_full}
Assume that the conditions \Cref{asmp:f_reg} and \Cref{asmp:tail} hold. Then, for $\gamma \in (-1,0)$, the following bound holds:
\begin{align}
\bigl|\D^\gamma f_\pi(x) - \Delta_{h,K}^{\gamma} f_\pi(x) \bigr| = {\cal O}\bigl(h^2 + 1/(hK)\bigr), \label{eqn:bound_deltah}
\end{align}
as $h$ goes to zero.
\end{thm}
Theorem~\ref{thm:riesz_full} shows that the error induced by our numerical approximation scheme is bounded and can be made arbitrarily small by decreasing $h$ and increasing $K$. 
We can also observe that for fixed $K$, the optimal $h = {\cal O}(K^{-1/3})$. 

The hidden constant in the right hand side of \eqref{eqn:bound_deltah} is allowed to depend on $x$ and let it be denoted as $C(x)$. In order to ease the analysis, in the rest of the paper we will assume that $\sup_x C(x)/\phi(x) < \infty$, so that Theorem~\ref{thm:riesz_full} can be directly used for bounding the error $|b- \tilde{b}|$ for any $x$. Note that this a mild assumption and holds trivially when $\x$ belongs to a bounded domain (e.g.\ the setting in \cite{wang2015privacy}). 

We now consider the following Euler-Maruyama discretization of \eqref{eqn:levysde} with the approximate drift:
\begin{align}
\tilde{\x}_{n+1} = \tilde{\x}_{n} + \eta_{n+1} \bta{\tilde{\x}_{n}} + \eta_{n+1}^{1/\alpha} \Delta L^\alpha_{n+1}, \label{eqn:euler_approx}
\end{align}
where the corresponding estimator is defined as: $\tilde{\nu}_N(g) \triangleq (1/H_N) \sum_{n=1}^N \eta_n g(\tilde{\x}_n)$. Even in this approximate Euler-scheme, we still obtain ULA as a special case of \eqref{eqn:euler_approx}, as we have $\lim_{\alpha \rightarrow 2} \bta{x} = -\partial_x U(x)$.

As opposed to $\bar{\nu}_N(g)$, $\tilde{\nu}_N(g)$ does not converge to $\nu(g)$ due to the error induced by the numerical approximation. However, fortunately, the weak error of this Euler scheme can still be bounded, as we show in the following theorem. For this result, we need an additional ergodicity condition\footnote{Proving the ergodicity of the SDEs in \Cref{assumption:ergo} is beyond the scope of this study; more information can be found in \cite{masuda2007ergodicity}. }. 
\begin{assumption}
\label{assumption:ergo}
The SDE \eqref{eqn:levysde} and $dX_t = \bta{X_t}dt + dL^\alpha_t$ are geometrically ergodic with their unique invariant measures.
\end{assumption}
\begin{thm}
\label{thm:euler_approx}
Assume that the conditions \Cref{asmp:stepsize,asmp:f_reg,assumption:ergo} hold, \Cref{asmp:lyapunov} holds for $\tilde{b}$,
and $K$ is chosen in such a way that \Cref{asmp:tail} holds for any $x$. Finally, assume that $|\partial_x g|$ is bounded. Then, the following bound holds almost surely:
\begin{align}
\bigl|\nu(g) - \lim\nolimits_{N \rightarrow \infty} \tilde{\nu}_N(g)\bigr| = {\cal O}\bigl(h^2 + 1/(hK) \bigr).
\end{align}
\end{thm}
This theorem shows that the weak error of the discretization in \eqref{eqn:euler_approx} is dominated by the numerical error induced by $\tilde{b}$ and can be made arbitrarily small by tuning $h$ and $K$.

\subsection{Multidimensional case}

Even though we have focused on the scalar case (i.e.\ $D=1$) so far, we can generalize the presented results to vector processes by using the same proof strategies since the components of $L^\alpha_t$ are independent\footnote{While extending our results to $D>1$, the independence of the components of $L^\alpha_t$ turns out to be a crucial requirement since the spectral measure of the corresponding multivariate stable distribution becomes discrete \cite{nolan2008overview}. Our results cannot be directly extended to SDEs that are driven by other multivariate stable processes, such as isotropic stable processes \cite{nolan2013multivariate}.}.  For $D>1$, the drift turns out to be a multidimensional generalization of \eqref{eqn:drift} and has the following form: (for $d = 1,\dots,D$)
\begin{align}
[\ba{x}]_d = \D^{\alpha-2}_{x_d} \{-\phi(x)  \partial_{x_d} U(x) \}/\phi(x), \label{eqn:drift_nd}
\end{align}
where $[v]_i$ denotes the $i$'th component of a vector $v$ and ${\cal D}^{\alpha}_{x_d}$ denotes the \emph{partial} fractional Riesz derivative along the direction $x_d$ \cite{ortigueira2014}. 
With this definition of the multidimensional drift, similar to the scalar case, we obtain the classical Langevin equation as a special case of \eqref{eqn:levysde}, since $\lim_{\alpha \rightarrow 2} \ba{x} =  - \nabla U(x)$ for $D>1$.

In applications, we can approximate \eqref{eqn:drift_nd} by applying the same numerical technique presented in \eqref{eqn:approxriesz} to each dimension $d$. However, for large $D$, this approach would be impractical since it would require the fractional derivatives to be computed $D$ times at each iteration. 

In this section, we propose a second scheme for approximating the fractional Riesz derivatives. The current approach is a computationally more efficient variant of the first numerical scheme presented in \eqref{eqn:delta_h} and it is given as follows: 
${\cal D}^\gamma f_\pi(x) \approx g_{\gamma,0} f_\pi(x)$, where $g_{\gamma,0} = \Gamma(\gamma+1)/\Gamma(\frac{\gamma}{2}+1)^2$ for $x \in \mathds{R}$.
In other words, we approximate the fractional derivatives by using only the first term of the centered difference operator defined in \eqref{eqn:delta_h}. When all the partial fractional derivatives in \eqref{eqn:drift_nd} are approximated with this approach, the multidimensional drift greatly simplifies and has the following form: (for $D>1$)
\begin{align}
\ba{\x}\approx \bha{\x} \triangleq  - c_\alpha \nabla U(\x), \label{eqn:drift_final}
\end{align}
where $c_\alpha \triangleq {\Gamma(\alpha-1)}/{\Gamma(\alpha/2)^2}$. We finally consider a discretization of \eqref{eqn:levysde} where the drift is approximated by \eqref{eqn:drift_final} and ultimately propose the Fractional Langevin Algorithm (FLA), defined as follows:
\begin{align}
\hat{\x}_{n+1} = \hat{\x}_{n} - \eta_{n+1} c_\alpha \nabla U(\hat{\x}_n) + \eta_{n+1}^{1/\alpha} \Delta L^\alpha_{n+1}. \label{eqn:euler_approx_hat}
\end{align}
Similar to the previous discretization schemes given in \eqref{eqn:euler} and \eqref{eqn:euler_approx}, FLA generalizes ULA as well, since $\lim_{\alpha \rightarrow 2} \bha{x} = -\nabla U(x)$. Besides, FLA has the exact same computational complexity as ULA, since it only requires to compute $\nabla U$ and generate $\Delta L^\alpha_{n}$. 
Another interesting observation is that $c_\alpha$ increases as $\alpha$ decreases, implying that FLA tends to increase the `weight' of the gradient as the driving process becomes heavier-tailed.

We now present our last theoretical result where we analyze the approximation error of the simplified scheme for $D=1$, and present it as a corollary to Theorems~\ref{thm:riesz_full} and \ref{thm:euler_approx}.
\begin{cor}
\label{cor:riesz_simple}
Assume that the conditions \Cref{asmp:f_reg,asmp:tail} hold. Let $K_x \in \mathds{N}_+$ be a value that satisfies \Cref{asmp:tail} for a given $x$, and let 
$\rmd: \mathds{R} \rightarrow \mathds{R}_+$ be a function defined as:
\begin{align*}
\rmd(x) \triangleq \Bigl|\sum\nolimits_{k \in \llbracket -K_x,K_x \rrbracket \setminus 0 } \frac{g_{\gamma,k} f_\pi(x-kh)}{g_{\gamma,0} f_\pi(x)} +1 \Bigr|^{1/\gamma}. 
\end{align*}
Then, for $-1<\gamma<0$, the following bound holds:
\begin{align*}
\Bigl|\D^\gamma f_\pi(x) - g_{\gamma,0} f_\pi(x) \Bigr| = {\cal O}\bigl( \rmd^2(x) + 1/(\rmd(x) K_x) \bigr).
\end{align*} 
Furthermore, if \Cref{asmp:stepsize} holds, \Cref{asmp:lyapunov,assumption:ergo} hold for $\hat{b}$, and $|\partial_x g|$ is bounded, then, the following bound holds:
\begin{align*}
\bigl|\nu(g) - \lim_{N \rightarrow \infty} \hat{\nu}_N(g)\bigr| = {\cal O}\bigl( \rmd^2(x^*) + 1/(\rmd(x^*) K_{x^*}) \bigr)
\end{align*}
almost surely, where $x^* = \argmax_x [ \rmd^2(x) + 1/(\rmd(x) K_x)]$ and $\hat{\nu}_N(g) \triangleq (1/H_N) \sum_{n=1}^N \eta_n g(\hat{\x}_n)$.
\end{cor}
Here, the term $\rmd(x)$ plays a similar role as the parameter $h$ in \eqref{eqn:delta_h}. 
This corollary shows that the approximation quality of \eqref{eqn:drift_final} may vary depending on the particular $x$ where $\D^\gamma f_\pi(x)$ is evaluated, and depending on the values of $K_x$ and $\rmd(x)$, $\hat{b}$ might even provide more accurate approximations than $\tilde{b}$ does. As a result, we observe that the weak error is dominated by the largest numerical error induced by $\hat{b}$. 
On the other hand, even if $\hat{b}$ would have a higher approximation error when compared to $\tilde{b}$, we would expect that the scheme in \eqref{eqn:euler_approx_hat} to be better behaved than \eqref{eqn:euler_approx}, since it is less prone to numerical instability.

\subsection{Large-scale Bayesian posterior sampling}

In Bayesian machine learning, the target distribution $\pi$ is often chosen as the Bayesian posterior: $\pi(\x) = p(\x| \Y)$, where $\Y \equiv \{\Y_i\}_{i=1}^{N_\Y}$ is a set of observed i.i.d.\ data points. 
This choice of the target distribution imposes the following form on the potential energy: 
$
  U(\x) = - (\sum_{i=1}^{N_\Y} \log p (\Y_i|\x) + \log p (\x) )
$,
where $p(\Y_i|\x)$ is the likelihood function and $p(\x)$ is the prior distribution.

In large scale applications, $N_{\Y}$ becomes very large and therefore computing $\nabla U$ at each iteration can be computationally inhibitive. Inspired by the Stochastic Gradient Langevin Dynamics (SGLD) algorithm \cite{WelTeh2011a}, which extends ULA to large-scale settings, we extend FLA by replacing the exact gradients $\nabla U$ in \eqref{eqn:euler_approx_hat} with an unbiased estimator, given as follows:
\begin{align*}
\nabla \tilde{U}_{n}(\x) = -[\nabla \log p(\x)  + \frac{N_\Y}{N_\Omega} \sum_{i \in \Omega_n} \nabla \log p (\Y_i | \x)],
\end{align*}
where $\Omega_n \subset \{1,\dots,N_\Y \}$ is a random data subsample that is drawn with replacement at iteration $n$ and $N_\Omega \ll N_\Y$ denotes the number of elements in $\Omega$. We call the resulting algorithm Stochastic Gradient FLA (SG-FLA). Note that SG-FLA coincides with SGLD when $\alpha = 2$. We leave the convergence analysis of SG-FLA as a future work.  

\begin{figure}[t]
\centering
\hfill \includegraphics[width=0.985\columnwidth]{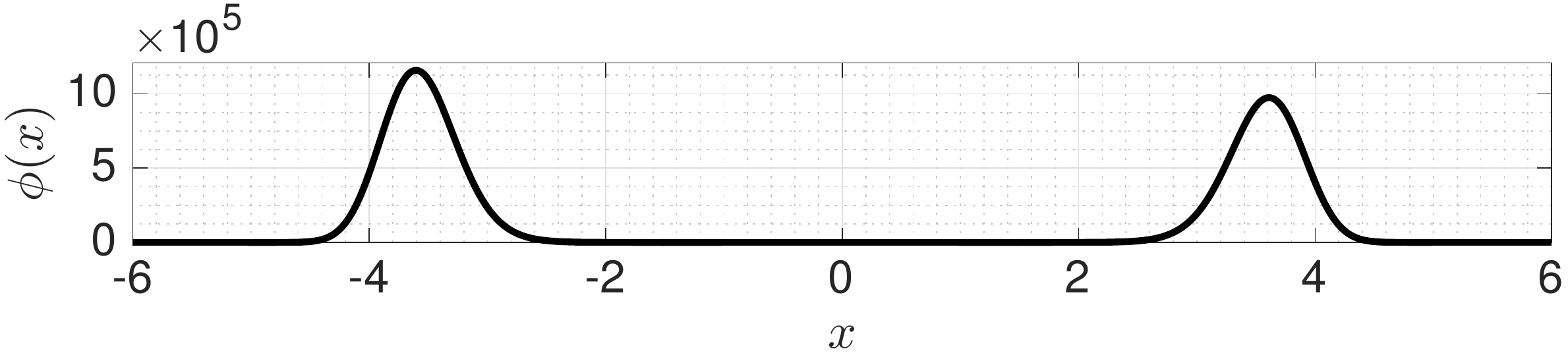} \\
\includegraphics[width=1\columnwidth]{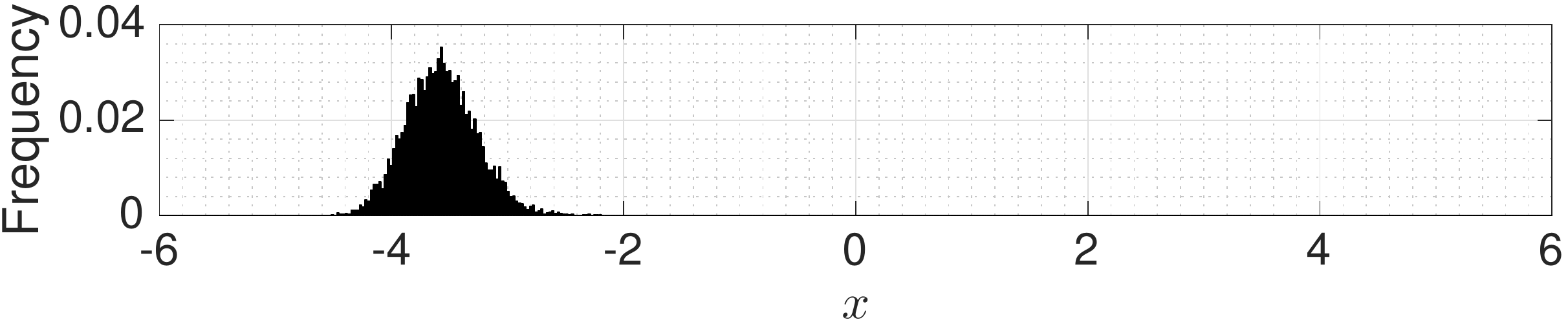} \\
\includegraphics[width=1\columnwidth]{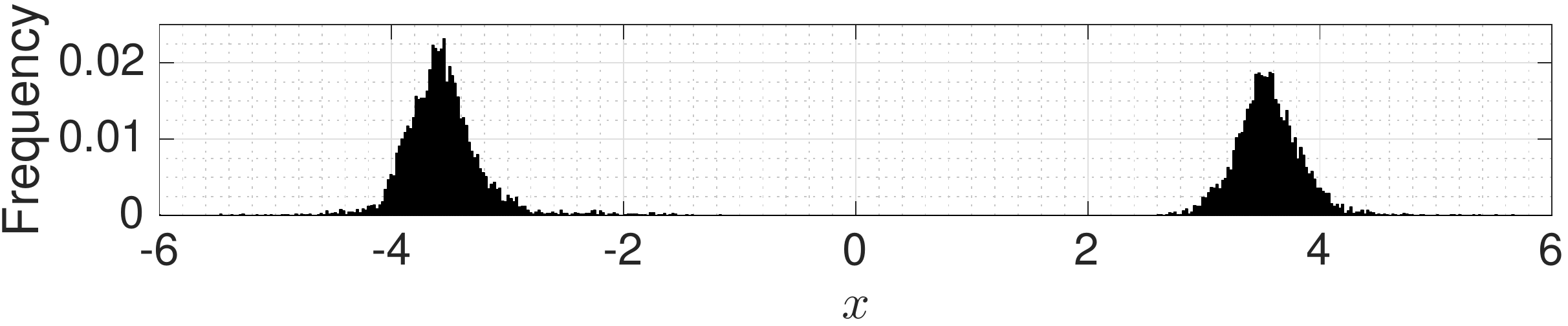}  
\vspace{-22pt}
\caption{Top: the double-well potential. Middle: the empirical distribution obtained via ULA (corresponds to FLA with $\alpha = 2$). Bottom: the empirical distribution obtained via FLA ($\alpha = 1.7$). }
\vspace{-3pt}
\label{fig:doublewell}
\end{figure}

\section{Experiments}

\textbf{The double-well potential: }
We conduct our first set of experiments on a synthetic setting where we consider the double-well potential, defined as follows: 
\begin{align*}
U(x) = (x+5)(x+1)(x-1.02)(x-5)/10+0.5.
\end{align*}
We illustrate the double-well potential in Figure~\ref{fig:doublewell} (top). It can be observed that the potential contains two well-separated modes with different heights, which makes the problem challenging. 

In our first experiment, we consider our first discretization scheme presented in \eqref{eqn:euler_approx}, where we approximate the true drift $b$, by $\tilde{b}$. Here, we use decreasing step-sizes that are determined as $\eta_n = (a_\eta/n)^{b_\eta}$, where we fix $a_\eta = 10^{-7}$ and $b_\eta = 0.6$. In each experiment, we generate $N=5000$ samples by using \eqref{eqn:euler_approx} and estimate the mean of the target distribution $\pi$ by using the sample average.   
For each $\alpha$, we run this scheme for different values of $h$ and $K$, repeat each experiment $5$ times, and monitor the  bias, where the ground truth is obtained via a numerical integrator.

We first fix $h=0.06$ and monitor the bias for increasing values of $K$. Here, we define the notion of an optimal $K$ as the smallest $K$, for whose larger values the performance improvement becomes negligible.
As we can observe from Figure~\ref{fig:doublewell_biasK} (top), the bias is gracefully degrading for increasing $K$, where the optimal $K$ depends on the choice of $\alpha$. We also observe that, modest values of $K$ seem sufficient for obtaining accurate results, especially for small $\alpha$.

In our second experiment, we fix $K=15$ and monitor the bias for different values of $h$. The results are illustrated in Figure~\ref{fig:doublewell_biasK} (bottom). We observe that the results support our theory: for very small values of $h$, the term $1/(hK)$ in the bound of Theorem~\ref{thm:euler_approx} dominates since $K$ is fixed. Therefore, we observe a drop in the bias as we increase $h$ up to a certain point, and then the bias gradually increases along with $h$. The results show that the performance becomes more sensitive to the value of $h$, as $\alpha$ becomes smaller.

\begin{figure}[t]
\centering
\includegraphics[width=\columnwidth]{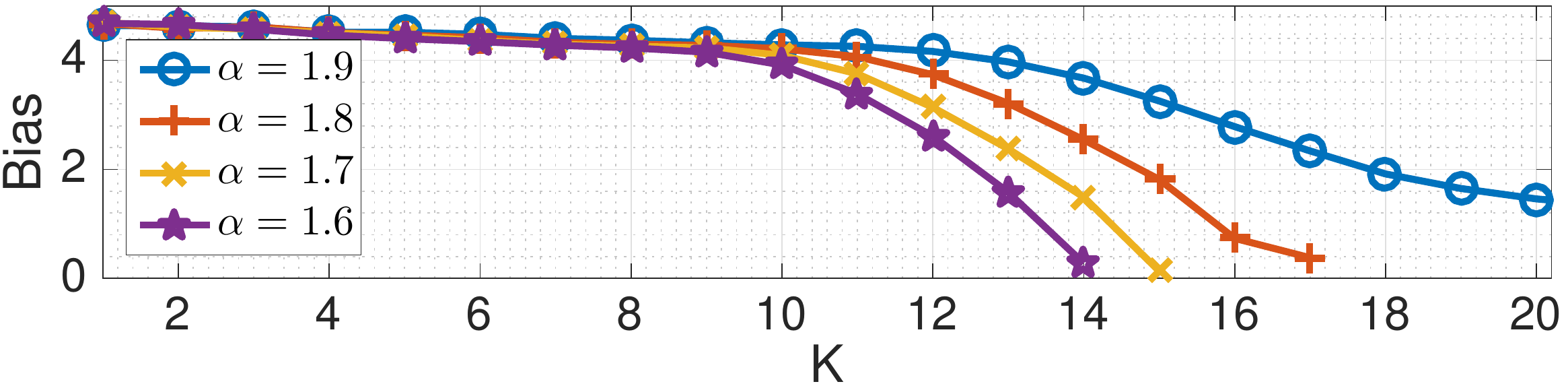} \\
\includegraphics[width=\columnwidth]{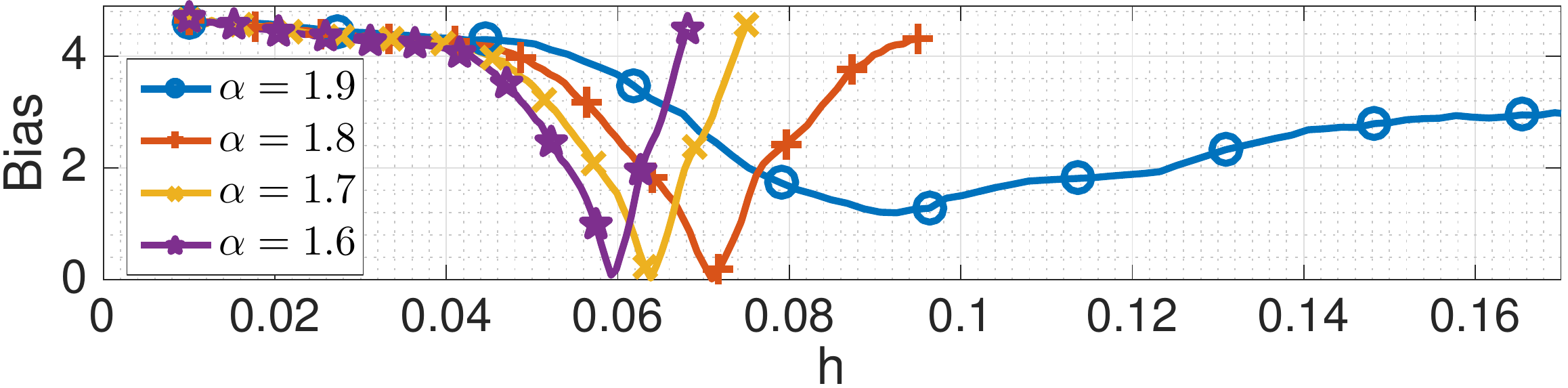} 
\vspace{-22pt}
\caption{Evaluation of the scheme in \eqref{eqn:euler_approx} on the estimation of the mean of the double-well potential. Top: the average bias vs $K$ for a fixed $h$. Bottom: the average bias vs $h$ for a fixed $K$.}
\vspace{-5pt}
\label{fig:doublewell_biasK}
\end{figure}

Even though the results in Figure~\ref{fig:doublewell_biasK} are promising, in practical applications we would not be able to use the scheme in \eqref{eqn:approxriesz} due to computational issues. 
In our next experiment, we aim to assess the approximation error of our second approximation scheme $\hat{b}$ given in \eqref{eqn:drift_final}.
However, the error $|\ba{x}-\bha{x}|$ cannot be measured in a straightforward manner, since $b$ cannot be computed exactly and the error itself depends on the particular point $x$.

Here, we develop an intuitive accuracy criterion for getting better insight into this error, where the aim is to compute the value of $K$ for which $\bta{x}$ and $\bha{x}$ would yield similar approximation errors on average. For a given $x$ and fixed $h$, we first choose a large enough $K^\star \in \mathds{N}_+$ and compute $b^\star(x,\alpha) \triangleq \tilde{b}_{h,K^\star}(x,\alpha)$ as our reference for $\ba{x}$. Then, for $K \in \llbracket 1, K^\star\rrbracket$, we compute the approximation error $e(x,\alpha,K) \triangleq | \bta{x} - b^\star(x,\alpha) |$ and the error induced by the ultimate approximation scheme: $\hat{e}(x,\alpha) \triangleq | \bha{x} - b^\star(x,\alpha) |$. We then find the value of $K$ for which $e(x,\alpha,K)$ and $\hat{e}(x,\alpha)$ are the closest: $\kappa(x,\alpha) = \argmin_K |e(x,\alpha,K) - \hat{e}(x,\alpha)|$.
We finally evaluate $\kappa$ on $I$ different points $\{x_i\}_{i=1}^I$ and use the average of these values as the measure of accuracy of $\hat{b}$, defined as:
$\hat{\kappa}(\alpha) = (1/I) \sum\nolimits_{i=1}^I \kappa(x_i,\alpha)$. Intuitively, this value is expected to be large when $\hat{b}$ yields a low error. 

In order to assess the accuracy of $\hat{b}$ in the double-well problem, we compute $\hat{\kappa}(\alpha)$ for different values of $\alpha$, where we fix $h=0.06$, $K^\star = 170$, and choose $\{x_i\}_i$ as $I = 200$ evenly-spaced points from the interval $[-5,5]$. The results are given in Table~\ref{tbl:kappa}. The results show that, despite its simplicity, $\hat{b}$ is able to provide reasonably accurate estimates for $b$. 
We observe that for $\alpha =1.6$, $\hat{\kappa}(\alpha)$ becomes $14.12$, which is even larger than the optimal $K$ for $\alpha =1.6$, as shown in Figure~\ref{fig:doublewell_biasK}.  
Therefore, it is promising to use $\hat{b}$ in real applications since it yields sufficiently accurate approximations with less computational requirements and does not require additional tuning for $h$ and $K$.

 \begin{table}[t]
 \label{tbl:kappa}
 \vspace*{-\baselineskip}
\caption{Evaluation of the accuracy of $\hat{b}$.}
\scalebox{0.9}
{
\begin{tabular}{c | c | c | c | c | c }
 \toprule
 & $\alpha = 1.5$ & $\alpha = 1.6$ & $\alpha = 1.7$ & $\alpha = 1.8$ & $\alpha = 1.9$\\
\midrule
$\hat{\kappa}(\alpha)$  & $19.31$ & $14.12$ & $12.72$ & $8.64$ & $7.03$ \\
\bottomrule
\end{tabular}
}
\vspace*{-0.1\baselineskip}
\end{table}

\begin{figure}[t]
\centering
\includegraphics[width=\columnwidth]{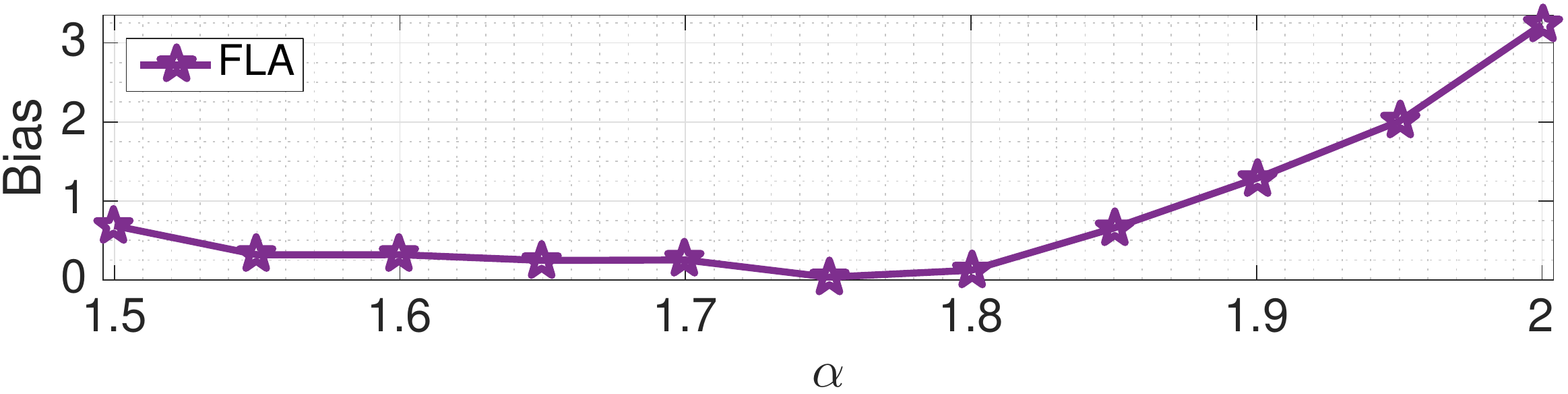} 
\vspace{-22pt}
\caption{The average bias obtained via FLA for different values of $\alpha$. Note that FLA corresponds to ULA when $\alpha = 2$. }
\vspace{-5pt}
\label{fig:doublewell_bias_bhat}
\end{figure}

In our last experiment on the double-well potential, we evaluate the ultimately proposed approach FLA on estimation of the mean of the target distribution. Similarly to the previous experiments, we run FLA for different values of $\alpha$, where we try several values for the hyper-parameters $a_\eta$ and $b_\eta$ and report the best results for each $\alpha$.
In each experiment, we generate $N=50000$ samples by using \eqref{eqn:euler_approx_hat} and repeat the procedure $10$ times.
We first illustrate two typical empirical distributions obtained via FLA and ULA in Figure~\ref{fig:doublewell} (middle, bottom). It can be clearly observed that ULA can locate only one of the modes, whereas FLA is able to locate both of the modes, thanks to the jumps of the $\alpha$-stable processes. This circumstance also reflects in the average bias, as illustrated in Figure~\ref{fig:doublewell_bias_bhat}. The results show that for $\alpha =2$ the average bias is around $3$, implying that the algorithm concentrates on either one of the modes at each trial, whereas we observe that the bias rapidly decreases as we decrease $\alpha$. The best performance is achieved when $\alpha = 1.75$. Finally we note that these results are also in line with the best-performing results given in Figure~\ref{fig:doublewell_biasK}.

\begin{figure}[t]
\centering
\includegraphics[width=0.94\columnwidth]{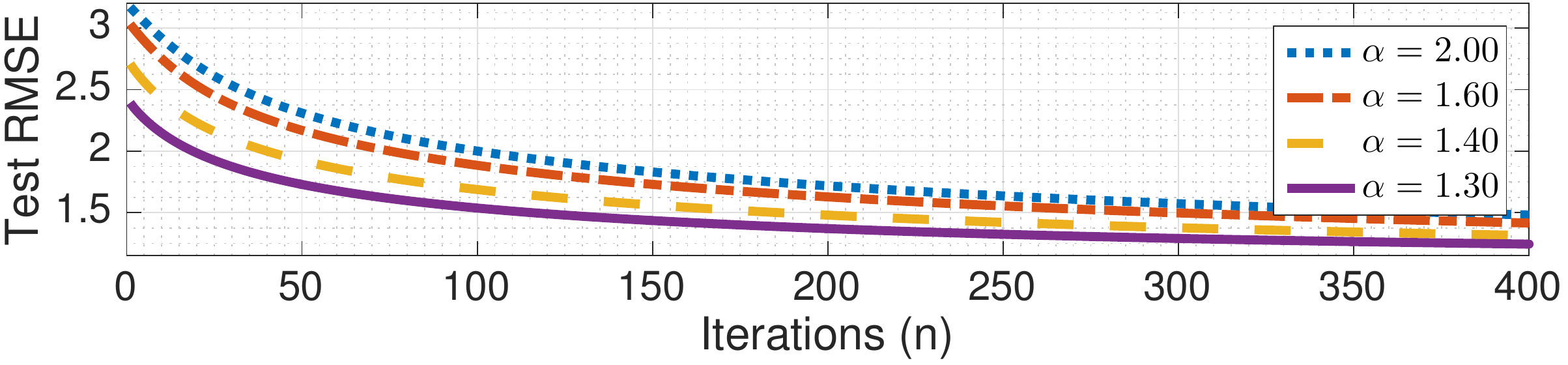} \\
\includegraphics[width=0.94\columnwidth]{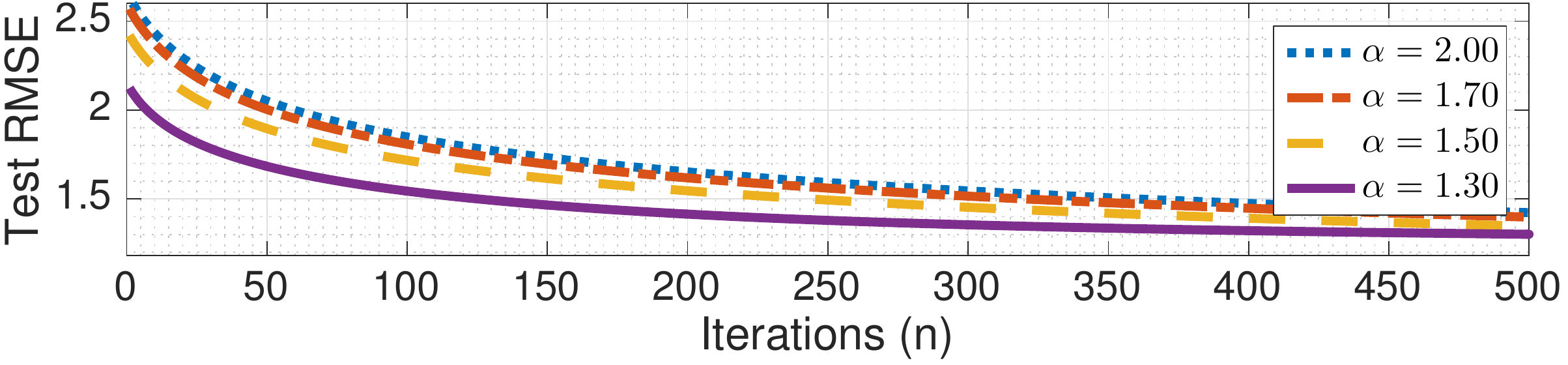}  \\
\includegraphics[width=0.94\columnwidth]{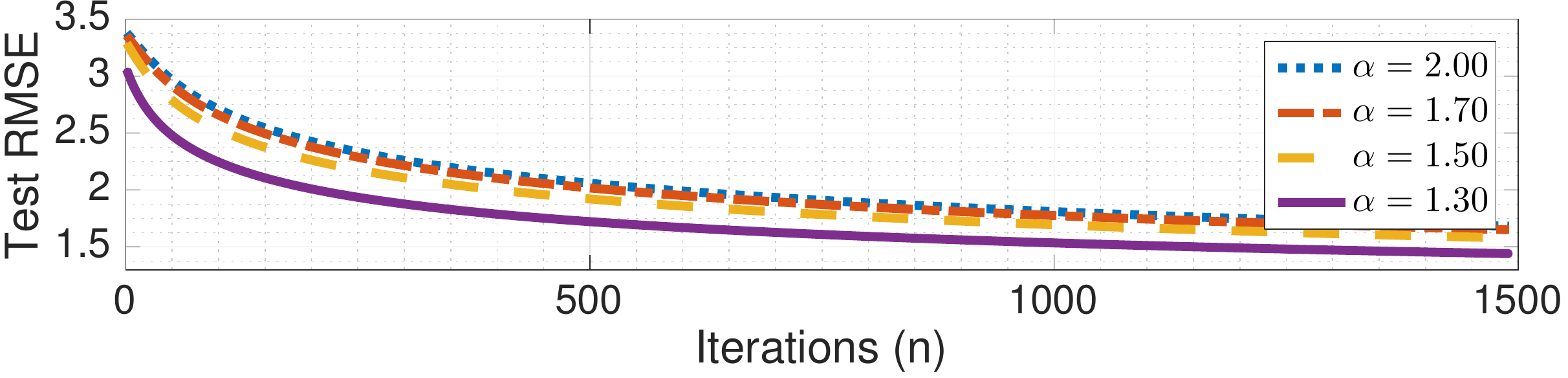}  
\vspace{-12pt}
\caption{The performance of SG-FLA on a link prediction application. Top: ML-$1$M, middle: ML-$10$M, bottom: ML-$20$M.}
\vspace{-10pt}
\label{fig:mf}
\end{figure}

\textbf{Matrix factorization: } 
In our second set of experiments, we switch to a large-scale Bayesian machine learning context.
We explore the use of SG-FLA on a large-scale link prediction application where we consider the following probabilistic  matrix factorization model \cite{gemulla2011,salakhutdinov2008bayesian}:
$
\nonumber A_{il} \sim {\cal N}(0,1), \> B_{lj} \sim {\cal N}(0,1), \> 
Y_{ij} | A,B \sim {\cal N}\bigl(\sum\nolimits_l A_{il} B_{lj}, 1 \bigr)
$,
where $Y \in \mathds{R}^{I \times L}$ is the observed data matrix with possible missing entries, and $A \in \mathds{R}^{I \times L}$ and $B \in \mathds{R}^{L \times J}$ are the latent factor matrices. The aim in this application is to predict the missing values of $Y$ by using a low-rank approximation. Recently, SGLD has been proven successful on similar models \cite{Ahn15,csimcsekli2015parallel,durmus2016stochastic,simsekli2017parallelized}. 

In this set of experiments, we apply SG-FLA on the three MovieLens movie ratings datasets (\url{grouplens.org}): MovieLens $1$Million (ML-$1$M), $10$Million (ML-$10$M), and $20$Million (ML-$20$M). 
The ML-$1$M dataset contains $1$ million non-zero entries, where $I = 3883$ (movies) and $J = 6040$ (users). The ML-$10$M dataset contains $10$ million non-zero entries, where $I = 10681$ and $J = 71567$. Finally, the ML-$20$M dataset contains $20$ million ratings, where $I = 27278$ and $J = 138493$.
In our experiments, we randomly select $10\%$ of the data as the test set and use the remaining data for generating the samples. The rank $L$ is set to $10$ for all datasets. In all experiments, we use decreasing step-sizes, where we fix $b_\eta = 0.51$ and try several values for $a_\eta$ and report the best results. We set $N_\Omega = N_\Y/10$ where $N_\Y$ denotes the number of non-zero entries in a given dataset.

Figure~\ref{fig:mf} shows the root mean squared-errors (RMSE) that are obtained on the three test sets. 
In all these experiments, we observe that the rate of convergence of SG-FLA increases as we decrease $\alpha$ from $2.0$ (i.e. SGLD) to $1.3$. In the case when $\alpha<1.3$, the jumps induced by the stable-L\'{e}vy motion becomes very large and the performance starts degrading. These results show that SG-FLA can be considered as a viable alternative to SGLD in large scale settings and it can provide improved performance over SGLD via minor algorithmic modifications, which come with the expense of tuning an additional parameter $\alpha$.

\textbf{Sigmoid Belief Networks: } 
In our last set of experiments, we investigate the use of SG-FLA on Sigmoid Belief Networks (SBN) \cite{gan2015learning}, which have been investigated in recent Stochastic Gradient MCMC studies \cite{chen2015convergence}. We make use of the software provided in \cite{chen2015convergence} and employ the identical experimental setup described therein: the binary observed data are assumed to be generated from a single binary hidden layer with sigmoid activations. The overall model is applied on the MNIST dataset, which contains $70$K binary images (of size $28 \times 28$) corresponding to different digits. 

In our experiments, we use an SBN with $100$ hidden units. We use a training set of $60$K images and $10$K images for testing, set the size of the data subsample $N_\Omega = 200$, and run SG-FLA for $5000$ iterations for training. Finally, we estimate the test likelihoods by using an annealed importance sampler \cite{salakhutdinov2008quantitative}.

As opposed to our previous experiments, we use constant step-sizes in these experiments, i.e.\ $\eta_n = \eta$ for all $n$, and investigate the performance of SG-FLA on SBNs for different values of $\eta$ and $\alpha$. The results are illustrated in Figure~\ref{fig:sbn}. We can observe that for small values of $\eta$, SG-FLA yields similar test likelihoods for all values of $\alpha$. However, as we increase the step size, we observe that the test likelihood of SGLD ($\alpha =2$) starts to diverge, whereas SG-FLA becomes more robust to large step sizes as $\alpha$ gets smaller. When $\alpha = 1.6$ the test likelihood stays almost constant for increasing values of $\eta$.
We do not observe an improvement in the performance for $\alpha<1.6$.

\begin{figure}[t]
\centering
\includegraphics[width=\columnwidth]{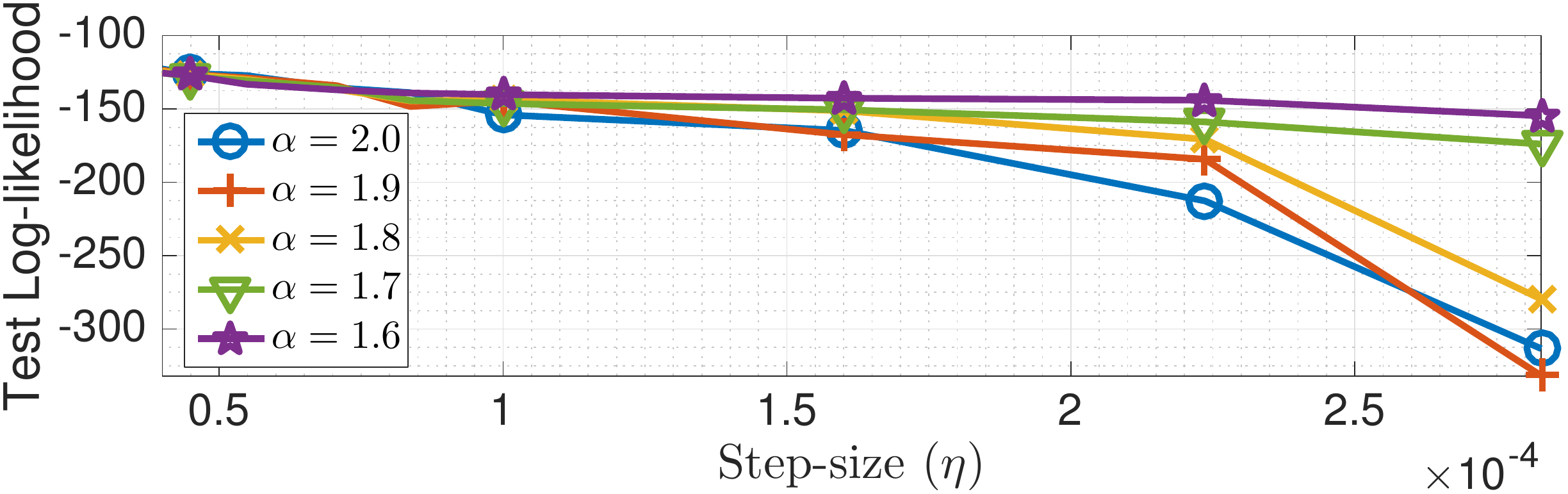} 
\vspace{-22pt}
\caption{The test log-likelihoods obtained via SG-FLA on SBNs, as a function of step-size $\eta$. }
\vspace{-5pt}
\label{fig:sbn}
\end{figure}

\section{Conclusion and Future Directions}
\label{sec:conc}

In this study, we explored the use of L\'{e}vy-driven SDEs within MCMC and presented a novel FLMC framework. We first showed that FLMC targets the correct distribution and then developed novel and scalable computational methods for practical applications. We provided formal analysis of the convergence properties and the approximation quality of the proposed numerical schemes. We supported our theory with several experiments, which showed that FLMC brings various benefits, such as providing superior performance in multi-modal settings, higher convergence rates, and robustness to algorithm parameters. 

The proposed framework opens up several interesting future directions: (i) the use of FLMC in simulated annealing for global optimization \cite{chen2016bridging}, where the jumps might bring further advantages (ii) extension of FLMC to `stable-like' processes \cite{bass1988uniqueness}, where $\alpha$ can depend on $\x$ (iii) incorporation of the local geometry for faster convergence \cite{PatTeh2013a,li2015preconditioned,simsekliICML2016} (iv) the use of SG-FLA in Bayesian model selection \cite{simsekli2016stochastic}.

\pagebreak

\section*{Acknowledgments}

The author would like to thank to Alain Durmus for his helps on the proofs, and to Roland Badeau, A. Taylan Cemgil, and Ga\"{e}l Richard for fruitful discussions.
The author would also like to thank to Changyou Chen for sharing the code used in the experiments conducted on SBNs.
This work is partly supported by the French National Research Agency (ANR) as a part of the FBIMATRIX project (ANR-16-CE23-0014), and the EDISON 3D project (ANR-13-CORD-0008-02).

\bibliography{levylangevin}
\bibliographystyle{icml2016}

%% file: supp_content.tex
\section{Numerically Stable Computation}

In this section, we focus on the computation of the following quantity:
\begin{align}
\frac{\D^{\gamma} \{ - \phi(x) \partial_x U(x) \} }{\phi(x)} \approx \frac1{h^\gamma} \sum_{k=-K}^K g_{\gamma,k}  \frac{- \phi(x-kh) \partial_x U(x-kh)}{\phi(x)}. \label{eqn:riesz_approx_num} 
\end{align}
Since $\phi(x) = \exp(-U(x))$, for very large values of $U(x)$ and $U(x-kh)$ we might easily end up with $0/0$ errors if we directly implement \eqref{eqn:approxriesz}.

\noindent We now present a numerically more stable algorithm for computing \eqref{eqn:approxriesz} . We rewrite the above equation as follows:
\begin{align}
\frac1{h^\gamma} \sum_{k=-K}^K g_{\gamma,k}  \frac{- \phi(x-kh) \partial_x U(x-kh)}{\phi(x)} &= \frac1{h^\gamma} \sum_{k=-K}^K g_{\gamma,k} \Bigl[-\partial_x U(x-kh) \exp\Bigl(\underbrace{U(x)-U(x-kh)}_{\ell_k}\Bigr) \Bigr] \\
&= \frac1{h^\gamma} \sum_{k=-K}^K g_{\gamma,k} \Bigl[ -\partial_x U(x-kh) \exp(\ell_k - \ell^* + \ell^*) \Bigr]\\
&= \frac{\exp \ell^* }{h^\gamma} \sum_{k=-K}^K g_{\gamma,k} \Bigl[ -\partial_x U(x-kh) \exp(\ell_k - \ell^*) \Bigr]
\end{align}
where $\ell^* = \max_{k \in \llbracket -K, K\rrbracket} \ell_k$. This numerical approach is similar to the well-known `log-sum-exp' trick.

\section{Proof of Theorem~\ref{thm:ffpe}}

Before proving Theorem~\ref{thm:ffpe}, we present the following proposition that will be helpful for our analysis.

\begin{prop}
\label{prop:riesz_der}
Let $f: \mathds{R} \rightarrow \mathds{R}$ be a differentiable function and assume that $\D^\gamma f(x)$ is well-defined for some $\gamma \in \mathds{R}$. Then, the following equality holds:
\begin{align}
\partial_x \D^\gamma f(x) = \D^\gamma \partial_x f(x).  
\end{align}  
\end{prop}
\begin{proof}
By definition we have:
\begin{align}
\D^\gamma f(x) &= {\cal F}^{-1} \{|\omega|^\gamma \hat{f}(\omega) \}, \\
\partial_x f(x) &= {\cal F}^{-1} \{ i \omega \hat{f}(\omega) \},
\end{align}
where ${\cal F}$ denotes the Fourier transform, $\hat{f}(\omega) = {\cal F}\{f(x)\} $, and $i = \sqrt{-1}$. By using these definitions, we obtain:
\begin{align}
\partial_x \D^\gamma f(x) &= {\cal F}^{-1} \{ {\cal F} \{ \partial_x \D^\gamma f(x)\}  \} \\
&= {\cal F}^{-1} \{  i \omega {\cal F} \{\D^\gamma f(x)\}  \} \\
&= {\cal F}^{-1} \{  i \omega |\omega|^\gamma \hat{f}(\omega)  \} \\
&= {\cal F}^{-1} \{  |\omega|^\gamma {\cal F} \{ {\cal F}^{-1} \{  i \omega  \hat{f}(\omega) \}  \} \} \\
&= {\cal F}^{-1} \{  |\omega|^\gamma {\cal F} \{ \partial_x f(x) \} \} \\
&= {\cal F}^{-1} \{  {\cal F} \{ \D^\gamma \partial_x f(x) \} \} \\
&=  \D^\gamma \partial_x f(x) .
\end{align}
This completes the proof.
\end{proof}

\subsection{Proof of Theorem~\ref{thm:ffpe}}

\begin{proof}
Let us define $q(\x,t)$ as the probability density function of the state $\x$ at time $t$. By Proposition 1 in \citepNew{schertzer2001fractional}, we obtain the fractional Fokker-Planck equation associated with the SDE given in \eqref{eqn:levysde} as follows:
\begin{align}
\partial_t q(\x,t) = -\partial_\x [ \ba{\x} q(\x,t)] -  \D^\alpha q(\x,t). 
\end{align}
By using the definition of $\ba{\x}$ we obtain
\begin{align}
\partial_t q(\x,t) &= -\partial_\x [ \frac{\D^{\alpha-2} \{-\phi(\x) \partial_\x U(\x) \} }{\phi(\x)} q(\x,t)] - \D^\alpha q(\x,t) \\
&= -\partial_\x [ \frac{\D^{\alpha-2} \{-\pi(\x) \partial_\x U(\x) \} }{\pi(\x)} q(\x,t)] - \D^\alpha q(\x,t).
\end{align}
Here, we used the fact that $\pi(\x) = \phi(\x)/Z$, where $Z = \int \phi(\x) d\x$. By using $-\partial_\x U(\x)  = \partial_\x \log \pi(\x) = \frac{\partial_\x \pi(\x)}{\pi(\x)}$, we obtain:
\begin{align}
\partial_t q(\x,t) &= -\partial_\x [ \frac{\D^{\alpha-2} \{\partial_\x \pi(\x) \} }{\pi(\x)} q(\x,t)] - \D^{\alpha} q(\x,t)
\end{align}
We can verify that $\pi(\x)$ is an invariant measure of the Markov process $(\x_t)_{t\geq 0}$ by checking
\begin{align}
-\partial_\x [ \frac{\D^{\alpha-2} \{\partial_\x \pi(\x) \} }{\pi(\x)} \pi(\x)] - \D^{\alpha} \pi(\x) &= -\partial_\x [ \D^{\alpha-2} \{\partial_\x  \pi(\x) \} ] - \D^{\alpha} \pi(\x) \\
&= -\partial^2_\x [ \D^{\alpha-2} \{ \pi(\x)\} ] - \D^{\alpha} \pi(\x) \label{eqn:interstep} \\
&= \D^2 [ \D^{\alpha-2} \{ \pi(\x) \}] - \D^{\alpha} \pi(\x)  \\
&= \D^{\alpha} \{\pi(\x)\} - \D^{\alpha} \{\pi(\x) \} \label{eqn:interstep2}\\
&= 0.
\end{align}
Here, we used the semigroup property of the Riesz potentials $\D^a \D^b f(x) = \D^{a+b} f(x)$ in \eqref{eqn:interstep2} and Proposition~\ref{prop:riesz_der} in \eqref{eqn:interstep}. If $\ba{\x}$ is Lipschitz continuous, by \citeNew{schertzer2001fractional} we can conclude that $\pi(\x)$ is the unique invariant measure of the Markov process $(\x_t)_{t\geq 0}$. 

\end{proof}

\section{Proof of Corollary~\ref{cor:euler}}

\begin{proof}
By Theorem~\ref{thm:ffpe}, we know that $\pi(\x)$ is the unique invariant distribution of the Markov process $(\x_t)_t$. Then, the claim directly follows Theorem 2 of \citepNew{panloup2008recursive}, provided that there exists $p\in (0,1/2]$ and $q\in[1/2,1]$, such that the following conditions hold:
\begin{align}
\int_{|x|>1} v(x) |x|^{2p} < \infty, \quad  \text{and} \quad \int_{|x|\leq1} v(x) |x|^{2q} < \infty,
\end{align}
where $v(x)$ is the L\'{e}vy-measure of the symmetric $\alpha$-stable L\'{e}vy process, defined as
\begin{align}
v(x) = \frac1{|x|^{\alpha+1}}. \label{eqn:levymeasure}
\end{align}
It is easy to see that these conditions hold with $p\in (0,1/2]$ and $q \in (\alpha/2,1]$. Therefore, we can directly apply Theorem 2 of \citepNew{panloup2008recursive} in order to obtain the desired result. 
\end{proof}

\section{Proof of Theorem~\ref{thm:riesz_full}}

Before proving Theorem~\ref{thm:riesz_full}, we first bound $|\D^\gamma f_\pi(x) - \Delta_h^{\gamma} f_\pi(x)|$ and $|\Delta_h^{\gamma} f_\pi(x) - \Delta_{h,K}^{\gamma} f_\pi(x)|$, which will be useful in our analysis.

\citetNew{ccelik2012crank} showed that $|\D^\gamma f(x) - \Delta_h^\gamma f(x)| = {\cal O}(h^2)$ for $1<\gamma\leq2$. However, we cannot directly use their result. For completeness, we adapt the proof of Lemma 2.2 in \citepNew{ccelik2012crank}, and prove that we obtain a bound of the same order for $-1<\gamma<0$. 

\begin{lemma}
\label{lem:riesz}
Assume $f(x) \in {\cal C}^3 (\mathds{R})$ and all derivatives up to order three belong to ${\cal L}_1(\mathds{R})$. Let $\Delta_h^\gamma$ be the operator defined in \eqref{eqn:delta_h}. Then, for $-1<\gamma<0$, the following bound holds:
\begin{align}
|\D^\gamma f(x) - \Delta_h^\gamma f(x)| = {\cal O}(h^2)
\end{align}
as $h$ goes to zero. 
\end{lemma}

\begin{proof}
We follow the same proof technique given in \citeNew{ccelik2012crank}. We first make use of the generator of \eqref{eqn:delta_h} given as follows: \citepNew{ortigueira2006riesz} 
\begin{align}
|2 \sin(z/2)|^\gamma = \sum_{k=-\infty}^\infty g_{\gamma,k} \exp(ikz).
\end{align}
Now, consider the Fourier transform of $\Delta_h^\gamma f(x)$
\begin{align}
{\cal F}\{\Delta_h^\gamma f(x)\} = \sum_{k=-\infty}^\infty g_{\gamma,k} \exp(ikh\omega) \hat{f}(\omega)
\end{align}
where ${\cal F}\{f(x)\} = \int_{-\infty}^\infty f(x) \exp(i\omega x) dx$ and $\hat{f}(\omega) \triangleq {\cal F}\{f(x)\}$. Then, we have 
\begin{align}
{\cal F}\{\Delta_h^\gamma f(x)\} = |2 \sin \frac{\omega h}{2}|^\gamma \hat{f}(\omega)
\end{align}
Let us define $\hat{\varphi}(h,\omega) = |\omega|^\gamma (1- \frac{|2 \sin \frac{\omega h}{2}|^\gamma}{|\omega h|^\gamma})\hat{f}(\omega)$. Then we have
\begin{align}
-\frac1{h^\gamma}{\cal F}\{\Delta_h^\gamma f(x)\} = -|\omega|^\gamma \hat{f}(\omega) + \hat{\varphi}(h,\omega). \label{eqn:delta_fourier}
\end{align}
Let us define $z = \omega h$ and $v(z) = \frac{|2\sin(z/2)|^\gamma}{|z|^\gamma}$. Now, we will bound the function $v(z)$. By using a Taylor expansion, we obtain
\begin{align}
v(z) &= \Bigl|\frac{2}{z}\Bigr|^\gamma \Bigl| \frac{z}{2} - \Bigl(\frac{z}{2}\Bigr)^3 \frac1{3!} +  \Bigl(\frac{z}{2}\Bigr)^5 \frac1{5!}  - \cdots \Bigr|^\gamma \\
&= \Bigl| 1 - \Bigl(\frac{z}{2}\Bigr)^2 \frac1{3!} +  \Bigl(\frac{z}{2}\Bigr)^4 \frac1{5!} - \cdots \Bigr|^\gamma 
\end{align}
Since $\gamma<0$, for small enough $z$, we have
\begin{align}
v(z) &\leq  \Biggl(1 - \Bigl|  \Bigl(\frac{z}{2}\Bigr)^2 \frac1{3!} +  \Bigl(\frac{z}{2}\Bigr)^4 \frac1{5!} + \cdots \Bigr| \Biggr)^\gamma \\
&= 1 - \gamma \Bigl|  \Bigl(\frac{z}{2}\Bigr)^2 \frac1{3!} +  \Bigl(\frac{z}{2}\Bigr)^4 \frac1{5!} + \cdots \Bigr| + \gamma (\gamma -1 ) \Bigl|  \Bigl(\frac{z}{2}\Bigr)^2 \frac1{3!} +  \Bigl(\frac{z}{2}\Bigr)^4 \frac1{5!} + \cdots \Bigr|^2 - \cdots \\
&\leq 1 + C_0 z^2 \\
&= {\cal O}(1+z^2).
\end{align}
By our assumptions, we have 
\begin{align}
|\hat{f}(\omega)| \leq C_1 (1 + |\omega|)^{-3}.
\end{align}
Therefore, we obtain
\begin{align}
|\hat{\varphi}(h,\omega)| &= |\omega|^\gamma |v(\omega h)-1| |\hat{f}(\omega)| \\
&\leq |\omega|^\gamma C_0 |\omega h|^2 C_1(1+|\omega|)^{-3}\\
&\leq C_2 h^2 (1+ |\omega|)^{\gamma+2}(1+|\omega|)^{-3}\\
&= C_2 h^2  (1+ |\omega|)^{\gamma-1}.
\end{align}
Since $-1<\gamma<0$, the inverse Fourier transform of $\hat{\varphi}(h,\omega)$ exists. Then we consider the inverse Fourier transform of \eqref{eqn:delta_fourier} and obtain
\begin{align}
-\frac1{h^\gamma} \Delta_h^\gamma f_{\pi}(x) = \D^\gamma f(x) + \varphi(h,x)
\end{align}
where 
\begin{align}
\varphi(h,x) \triangleq {\cal F}^{-1}\{\hat{\varphi}(h,\omega)\} = \frac1{2 \pi} \int_{-\infty}^\infty \hat{\varphi}(h,\omega) \exp(-i\omega x) d\omega .
\end{align}
By using the bound for $|\hat{\varphi}(h,\omega)|$, we obtain
\begin{align}
|\varphi(h,x)| &= \frac1{2\pi} \int_{-\infty}^\infty \hat{\varphi}(h,\omega) \exp(-i\omega x) d\omega \\
&\leq \frac1{2\pi} \int_{-\infty}^\infty |\hat{\varphi}(h,\omega)|  d\omega \\
&\leq \frac1{2\pi} \int_{-\infty}^\infty  C_2 h^2  (1+ |\omega|)^{\gamma-1}  d\omega \\
&\leq C_3 h^2.
\end{align}
Finally, we conclude that
\begin{align}
|\D^\gamma f(x) -\frac1{h^\gamma} \Delta_h^\gamma f_{\pi}(x)| &=  |\varphi(h,x)| \\
&\leq C_3 h^2 .
\end{align}

\end{proof}

Now, we bound the term $|\Delta_h^{\gamma} f_\pi(x) - \Delta_{h,K}^{\gamma} f_\pi(x) |$.

\begin{lemma}
\label{lem:riesz_trunc}
Assume $\Bigl| f_\pi(x-kh)\Bigr| \leq C \exp(-|k|h) $ for some $C >0$ and $|k| > K$ for some $K<\infty$, where $K \in \mathds{N}_+$. Then the following bound holds:
\begin{align}
\Bigl|\Delta_h^{\gamma} f_\pi(x) - \Delta_{h,K}^{\gamma} f_\pi(x) \Bigr| = {\cal O}\Bigl(\frac{1}{hK}\Bigr).
\end{align}
\end{lemma}
\begin{proof}
By definition we have
\begin{align}
\Bigl|\Delta_h^{\gamma} f_\pi(x) - \Delta_{h,K}^{\gamma} f_\pi(x) \Bigr| &= \Bigl| h^{-\gamma} \sum_{k \notin \llbracket -K, K\rrbracket}  g_{\gamma,k} f_\pi(x-kh)  \Bigr| \\
&\leq  h^{-\gamma} \sum_{k \notin \llbracket -K, K\rrbracket}  g_{\gamma,k} \Bigl|f_\pi(x-kh)\Bigr| 
\end{align}
By the hypothesis and the symmetry of the coefficients ($g_{\gamma,k} = g_{\gamma, -k}$), we have
\begin{align}
\Bigl|\Delta_h^{\gamma} f_\pi(x) - \Delta_{h,K}^{\gamma} f_\pi(x) \Bigr| &\leq C h^{-\gamma} \sum_{k = K+1}^\infty  g_{\gamma,k} \exp(-k h) 
\end{align}
From \citepNew{ortigueira2006riesz,ccelik2012crank}, we know that $g_{\gamma,k} = {\cal O}(\frac1{k^{\gamma+1}})$, then we obtain
\begin{align}
\Bigl|\Delta_h^{\gamma} f_\pi(x) - \Delta_{h,K}^{\gamma} f_\pi(x) \Bigr| &\leq C h^{-\gamma} \sum_{k = K+1}^\infty  \frac1{k^{\gamma+1}} \exp(-k h)\\
&= C h  \sum_{k = K+1}^\infty  \frac1{(hk)^{\gamma+1}} \exp(-k h) \\
&\leq C h \int_{K}^\infty (yh)^{-(\gamma+1)} \exp(-yh) \> dy
\end{align}
By making a change of variables, we obtain
\begin{align}
\Bigl|\Delta_h^{\gamma} f_\pi(x) - \Delta_{h,K}^{\gamma} f_\pi(x) \Bigr| &\leq C  \int_{hK}^\infty y^{-(\gamma+1)} \exp(-y) \> dy \\
&= C \Gamma(-\gamma,hK),
\end{align}
where $\Gamma(\cdot,\cdot)$ denotes the incomplete gamma function \citepNew{borwein2009uniform}. Then by using Theorem 2.4 of \citeNew{borwein2009uniform}, we obtain the desired result as follows:
\begin{align}
\Bigl|\Delta_h^{\gamma} f_\pi(x) - \Delta_{h,K}^{\gamma} f_\pi(x) \Bigr| &\leq C \frac1{hK}.
\end{align}
\end{proof}

\subsection{Proof of Theorem~\ref{thm:riesz_full}}

\begin{proof}
We decompose the error as follows:
\begin{align}
\Bigl|\D^\gamma f_\pi(x) - \Delta_{h,K}^{\gamma} f_\pi(x) \Bigr| \leq \Bigl|\D^\gamma f_\pi(x) - \Delta_h^{\gamma} f_\pi(x) \Bigr| + \Bigl|\Delta_h^{\gamma} f_\pi(x) - \Delta_{h,K}^{\gamma} f_\pi(x) \Bigr|.
\end{align}
Then we obtain the desired result by applying Lemmas~\ref{lem:riesz} and \ref{lem:riesz_trunc} to the first and the second terms of the left hand side of the above inequality.
\end{proof}

\section{Proof of Theorem \ref{thm:euler_approx}}

Before presenting the proof of Theorem~\ref{thm:euler_approx}, let us define the following SDEs which will be useful in the analysis:
\begin{align}
d\x_t &= \ba{\x_{t-}}dt + dL^\alpha_t \label{eqn:sdex} \\
d\y_t &= \bta{\y_{t-}}dt + dL^\alpha_t\label{eqn:sdey} 
\end{align}
where $b$ and $\tilde{b}$ are defined in \eqref{eqn:drift} and \eqref{eqn:approxriesz}, respectively. Here, \eqref{eqn:sdex} is our main SDE, \eqref{eqn:sdey} is another SDE whose drift is $\tilde{b}$.  

Let us first present the following lemma, which will be useful for proving Theorem~\ref{thm:euler_approx}.

\begin{lemma}
\label{lem:weakerror}
Let $(\x_t)_{t\geq 0}$ and $(\y_t)_{t\geq 0}$ be the solution processes of the SDEs \eqref{eqn:sdex} and \eqref{eqn:sdey}. Assume that both $(\x_t)_{t\geq 0}$ and $(\y_t)_{t\geq 0}$ are geometrically ergodic with their unique invariant measures and $|\partial_x g|$ is bounded. Further assume that the truncation parameter $K$ is chosen in such a way that \Cref{asmp:tail} holds for any $x$. Then the following bound on the weak error holds:
\begin{align}
\Bigl| \mathds{E}[g(\x_t) - g(\y_t)] \Bigr| \leq C \bigl(1- \exp(- \lambda t)\bigr) \bigl(h^2 + \frac1{hK}\bigr) ,
\end{align}
for some $C,\lambda>0$.
\end{lemma}

\begin{proof} 
We follow a standard approach for weak error analysis in SDEs. We make use of the semigroups associated with $(\x_t)_{t\geq 0}$ and $(\y_t)_{t\geq 0}$, given as $\pb_t g(x)\triangleq \mathds{E}[g(\x_t)]$ and $\pt_t g(x) \triangleq  \mathds{E}[g(\y_t)]$.
Then, we rewrite the weak error by using the semigroups, given as follows \citeNew{kohatsu2015short}:
\begin{align}
 \mathds{E}[g(\x_t) - g(\y_t)] &= \pb_t g(x) - \pt_t g(x) \\
 &= \int_{0}^t \partial_s \{ \pb_s \pt_{t-s} g(x) \} ds. \label{eqn:weak_error_semigroup}
\end{align} 
We now investigate the integrand, as follows:
\begin{align}
\partial_s \{ \pb_s  \pt_{t-s} g(x) \} &=  (\partial_s \pb_s) \pt_{t-s} g(x) + \pb_s (\partial_s \pt_{t-s}) g(x) \\
&= \pb_s \ab \pt_{t-s} g(x) - \pb_s \at \pt_{t-s} g(x) \\
&= \pb_s \{ \ab -  \at \} \pt_{t-s} g(x), 
\end{align}
where $\ab$ and $\at$ are the generators of the SDEs in \eqref{eqn:sdex} and \eqref{eqn:sdey}, respectively, and they are defined as follows \citeNew{duan}: 
\begin{align}
\ab f(x) &\triangleq \ba{x} \partial_x f(x) + \int_{\mathds{R} \setminus 0 } [f(x+y) - f(x) - \mathds{1}_{\{|y|<1\}} y \partial_x f(x)]v(dy), \\
\at f(x) &\triangleq \bta{x} \partial_x f(x) + \int_{\mathds{R} \setminus 0 } [f(x+y) - f(x) - \mathds{1}_{\{|y|<1\}} y \partial_x f(x)]v(dy),
\end{align}
for a differentiable function $f$, where $\mathds{1}$ is the indicator function and $v(dy)$ is the L\'{e}vy-measure of the symmetric $\alpha$-stable L\'{e}vy process defined in \eqref{eqn:levymeasure}.
Since these SDEs have the same volatility, the difference $(\ab - \at)f(x)$ simplifies and it is equal to $(\ba{x}-{\bta{x}})\partial_x f(x)$. Accordingly, we obtain the following expression: 
\begin{align}
 \partial_s \{ \pb_s  \pt_{t-s} g(x) \} &= \pb_s  (\ba{x}-{\bta{x}})\partial_x \pt_{t-s} g(x), \\
 &=  \pb_s  (\ba{x}-{\bta{x}}) \pt_{t-s} \partial_x g(x) \label{eqn:sde_semigroup}
\end{align} 
where we assumed the interchangeability of integration and differentiation. 
By the ergodicity assumptions, we have: 
\begin{align}
|\pb_{s} f(x)| &\leq C \exp\bigl(-\lambda_x s\bigr) \|f\|_{\infty}, \label{eqn:ergo1} \\
|\pt_{t-s} f(x)| &\leq C \exp\bigl(-\lambda_y(t-s)\bigr) \|f\|_{\infty} \label{eqn:ergo2}
\end{align}
for some $C, \lambda_x, \lambda_y>0$ and a bounded function $f$. By injecting \eqref{eqn:sde_semigroup} into \eqref{eqn:weak_error_semigroup} and then
using the boundedness assumption on $\partial_x g$, \eqref{eqn:ergo1}, \eqref{eqn:ergo2}, and Theorem~\ref{thm:riesz_full}, we obtain the following inequality: (for some $C>0$)
\begin{align} 
\Bigl| \mathds{E}[g(\x_t) - g(\y_t)] \Bigr| &\leq C \Bigl(h^2 + \frac1{hK}\Bigr) \int_0^t  \exp(-\lambda_x s)ds \\
&\leq C \bigl(1- \exp(-\lambda_x t)\bigr) \Bigl(h^2 + \frac1{hK}\Bigr),
\end{align}
as desired. This completes the proof. 

\end{proof}

\subsection{Proof of Theorem \ref{thm:euler_approx}}

\begin{proof}
Let us first define the following quantities:
\begin{align}
\nu(g) &=  \int g(\x) \pi(d\x) \\
\tilde{\nu}(g) &=  \int g(\y) \tilde{\pi}(d\y) 
\end{align}
where $\pi$, and $\tilde{\pi}$ are the unique invariant measures of \eqref{eqn:sdex} and \eqref{eqn:sdey}, respectively. And let $(\x_t)_{t\geq 0}$ and $(\y_t)_{t\geq 0}$ be the solution processes of the SDEs \eqref{eqn:sdex} and \eqref{eqn:sdey}. By the triangle inequality, we have 
\begin{align}
\Bigl|\nu(g) - \lim_{N \rightarrow \infty} \tilde{\nu}_N(g) \Bigr| &\leq \Bigl|\nu(g) - \tilde{\nu}(g)\Bigr| + \Bigl|\tilde{\nu}(g) -  \lim_{N \rightarrow \infty} \tilde{\nu}_N(g) \Bigr| .
\end{align}
Due to the ergodicity assumptions, we can rewrite the right hand side of the above inequality as follows:
\begin{align}
\Bigl|\nu(g) - \lim_{N \rightarrow \infty} \tilde{\nu}_N(g)\Bigr| &\leq \Bigl| \lim_{t \rightarrow \infty} \mathds{E}\bigl[g(\x_t) - g(\y_t)\bigr]  \Bigr| + \Bigl|\tilde{\nu}(g) -  \lim_{N \rightarrow \infty} \tilde{\nu}_N(g) \Bigr| \\
 &= \lim_{t \rightarrow \infty} \Bigl| \mathds{E}\bigl[g(\x_t) - g(\y_t)\bigr] \Bigr| + \Bigl|\tilde{\nu}(g) -  \lim_{N \rightarrow \infty} \tilde{\nu}_N(g) \Bigr| \label{eqn:squeeze}
\end{align}
where \eqref{eqn:squeeze} can be obtained by the reverse triangle inequality and the squeeze theorem. 
By \citeNew{panloup2008recursive}, we have almost surely
\begin{align}
\Bigl| \tilde{\nu}(g) -  \lim_{N \rightarrow \infty} \tilde{\nu}_N(g) \Bigr| = 0 . \label{eqn:lim3} 
\end{align}
By Lemma~\ref{lem:weakerror}, we have
\begin{align}
\lim_{t \rightarrow \infty} \Bigl| \mathds{E}\bigl[g(\x_t) - g(\y_t)\bigr] \Bigr| \leq C (h^2 + \frac1{hK}), \label{eqn:lim1}
\end{align}
for some $C>0$. 
Finally, by injecting \eqref{eqn:lim3}, and \eqref{eqn:lim1} in \eqref{eqn:squeeze}, we obtain the desired result:
\begin{align}
|\nu(g) - \lim_{N \rightarrow \infty} \tilde{\nu}_N(g)| \leq C (h^2 + \frac1{hK}) . 
\end{align}
This completes the proof. 

\end{proof}

\begin{remark}
The assumption \Cref{assumption:ergo} is not very restrictive and for the SDE \eqref{eqn:levysde} it can be easily satisfied if the following conditions hold:

A1) $x b(x,\alpha) <= - a x^2 +c$ with $a,c>0 $ .\\
A2) Let $S_K^i (x) = \sum_{k=-K}^K g_{k,\gamma} f_\pi^{(i)}(x-kh)$, where $f_\pi^{(i)}$ is the $i$'th derivative of $f_\pi$. Then $\{S_K\}_{K > 0}$ converges uniformly when $K \rightarrow \infty$ and is bounded, for any $x$ belonging to a compact set and $i >= 1$. 

A1 is a standard growth condition. A2 is mild due to the nature of $f_\pi$, and it ensures $X_t$ to have a smooth density (see \citeNew{picard1996existence}). We note that, if the SDE \eqref{eqn:levysde} satisfies A1-2, then it is easy to show that so does the perturbed SDE defined in \Cref{assumption:ergo}. 
\end{remark}

\section{Proof of Corollary~\ref{cor:riesz_simple}}

\begin{proof}
It is easy to check that \eqref{eqn:drift_final} corresponds to using \eqref{eqn:approxriesz} with $h= \rmd(x)$. Then we obtain the first part of the conclusion by directly applying Theorem~\ref{thm:riesz_full}. The second part of the conclusion can be proved by using the same proof technique presented in Theorem~\ref{thm:euler_approx}.
\end{proof}

\begin{remark}
Corollary~\ref{cor:riesz_simple} implies that the weak error FLA depends heavily on the structure of the target density. If the high probability region of $\pi$ is concentrated in a particular area, $K_x$ would be small and vice versa. On the other hand, if $x$ is near a mode of $\pi$ or $f_\pi(x)$ is symmetric around $x$, or $f_\pi$ varies very slowly with $x$, $r(x)$ can be arbitrarily small. Finally, Corollary~\ref{cor:riesz_simple} expresses the overall error in terms of $K_x$ and $r(x)$, and illustrates the roles of these terms.
\end{remark}

\section{A Note on the Experiments Conducted on SG-FLA}

In the SG-FLA experiments, we monitored the training likelihood and we did not observe that SG-FLA is able to find a better mode in a systematic way. However, we did observe that SG-FLA is more robust to the size of the minibatches -- therefore to the variance of the stochastic gradients -- when compared to SGLD. We believe that this observation is caused by the fact that the jumps in SG-FLA provide robustness against stochastic gradients and the choice of the step sizes.

\bibliographystyleNew{icml2016}
\bibliographyNew{levylangevin}